\documentclass[3p,authoryear,11pt]{elsarticle}

\usepackage{latexsym}
\usepackage{epsfig}
\usepackage{stmaryrd}
\usepackage{multicol}
\usepackage{pdfsync}
\usepackage{bm}
\usepackage{bbm}
\usepackage{dsfont}
\usepackage{comment}
\usepackage{caption}
\usepackage{psfrag}
\usepackage{xcolor}
\usepackage{graphics}
\usepackage{enumitem}
\usepackage{mathtools}
\usepackage{bbm}
\usepackage{mathrsfs} 
\usepackage{color}
\usepackage{hyperref}
\usepackage{cancel}
\usepackage{float}
\usepackage{multirow}
\usepackage{eurosym}
\usepackage{natbib}
\usepackage{titlesec}
\usepackage{subcaption}
\journal{TBA}
\usepackage[utf8]{inputenc}
\usepackage{pgfplots}

\usepgfplotslibrary{groupplots,dateplot}
\usetikzlibrary{patterns,shapes.arrows}
 \usetikzlibrary{decorations.pathreplacing}
\pgfplotsset{compat=newest}

\usepackage{amsmath,amsfonts,amssymb,amsthm}

\usepackage{algorithmic}

\usepackage[ruled,vlined,linesnumbered]{algorithm2e}

\usepackage{tikz-cd}
\usepackage{tikz}
\usetikzlibrary{decorations.pathreplacing,positioning, arrows}
\hypersetup{
	colorlinks = true, 
	urlcolor = blue, 
	linkcolor = red, 
	citecolor = purple 
}

\newcommand{\ourmodel}{PULSE\xspace}

\newcommand{\spread}{\mathcal{S}}
\newcommand{\imbalance}{\mathcal{I}}
\newcommand{\inventory}{\mathcal{Q}}
\newcommand{\cash}{\mathcal{C}}
\newcommand{\volatility}{\mathcal{V}}
\newcommand{\return}{\mathcal{R}}

\newcommand{\window}{\ell}

\newcommand{\valuetoxichorizon}{\mathcal{G}}
\newcommand{\toxichorizon}{toxicity horizon}

\newcommand{\cutoff}{\mathfrak{p}}
\newcommand{\cutoffprobability}{cutoff probability}

\newcommand{\N}{\mathbb{N}}

\newcommand{\clock}{\mathfrak{c}}

\newcommand*{\vertbar}{\rule[-1ex]{0.5pt}{2.5ex}}
\newcommand*{\horzbar}{\rule[.5ex]{2.5ex}{0.5pt}}

\newcommand{\E}{{\mathbb{E}}}

\newcommand{\yt}[1][t]{y_{#1}}
\newcommand{\xt}[1][t]{{\bf x}_{#1}}
\newcommand{\reals}{\mathbb{R}}
\newcommand{\normdist}[3]{\mathcal{N}\left( {#1} \,\vert\, {#2},\, {#3} \right)}
\newcommand{\data}{\mathcal{D}}
\newcommand{\plast}{{\bf w}}
\newcommand{\phidden}{{\boldsymbol\psi}}
\newcommand{\phiddensub}{{\bf z}}
\newcommand{\vdlast}[1][t]{{\phi_{#1}\left(\plast\right)}}
\newcommand{\vdhidden}[1][t]{{\varphi_{#1}\left(\phiddensub\right)}}
\newcommand{\covlast}{{\boldsymbol{\Sigma}}}
\newcommand{\covhidden}{{\boldsymbol{\Gamma}}}
\newcommand{\mhidden}{{\boldsymbol{\mu}}}
\newcommand{\mlast}{{\boldsymbol{\nu}}}

\newcommand{\dwarmup}{{{\cal D}_\text{warmup}}}

\newcommand{\normdistlast}{{\normdist{\plast}{\mlast}{\covlast}}}
\newcommand{\normdisthidden}{{\normdist{\phiddensub}{\mhidden}{\covhidden}}}

\newcommand{\mfT}{{\mathfrak{T}}}

\newcommand{\mcA}{{\mathcal{A}}}

\newcommand{\mfD}{{\mathfrak{D}}}

\newtheorem{theorem}{Theorem}
\newtheorem{corollary}{Corollary}

\newtheorem{proposition}{Proposition}

\newtheorem{definition}[theorem]{Definition}

\numberwithin{equation}{section}
\newcommand{\diff}{\mathrm{d}}

\DeclareMathOperator*{\argmax}{arg\,max}
\DeclareMathOperator*{\argmin}{arg\,min}

\title{\textbf{Detecting Toxic Flow}}
\tnotetext[label0]{We thank Andrew Stewart, Alistair Sturgiss, Fayçal Drissi, Patrick Chang, \'Alvaro Arroyo,
Sergio Calvo,
the participants at the Victoria Seminar Series, the Crossroads Seminar,
and the University of Edinburgh for comments. 
ChatGPT suggested the name PULSE for our algorithm.
}

\author[label1,label2]{\'{A}lvaro Cartea}

\address[label1]{Mathematical Institute, University of Oxford, Oxford, UK}
\address[label2]{Oxford-Man Institute of Quantitative Finance, Oxford, UK}
\ead{alvaro.cartea@maths.ox.ac.uk}

\author[label2]{Gerardo Duran-Martin}

\author[label1,label2]{Leandro S\'{a}nchez-Betancourt}
\ead{sanchezbetan@maths.ox.ac.uk}

\begin{document}

\begin{abstract}
This paper develops a framework to predict toxic trades that a broker receives from her clients.
Toxic trades are predicted with a novel online learning Bayesian method which we call the
\textit{projection-based unification of last-layer and subspace estimation} ({\ourmodel}).
{\ourmodel} is a
fast and statistically-efficient Bayesian procedure for online training of neural networks.
We employ a proprietary dataset of foreign exchange transactions to test our methodology. Neural networks trained with {\ourmodel}
outperform standard machine learning and statistical methods when predicting if a trade will be toxic;
the benchmark methods are logistic regression, random forests, and a recursively-updated maximum-likelihood estimator.
We devise a strategy for the broker who uses toxicity predictions to internalise or to externalise each trade received from her clients. 
Our methodology can be implemented in real-time because it takes less than one millisecond to update parameters and make a prediction.
Compared with the benchmarks,
online learning of a neural network with {\ourmodel}
attains the highest PnL and avoids the most losses by externalising toxic trades. 

\end{abstract}

\maketitle
\section{Introduction}
Liquidity  providers are key to well-functioning financial markets. In foreign exchange (FX), as in other asset classes,
broker-client relationships are ubiquitous. The broker streams bid and ask quotes to her clients and the clients decide when to trade on these quotes, so the broker bears the risk of adverse selection when trading with better informed clients. 
These risks are borne by both liquidity providers who stream quotes to individual parties and by market participants who provide liquidity in the books of electronic exchanges. However, in contrast to electronic order books in which trading is anonymous for all participants (e.g., Nasdaq, LSE, Euronext), in broker-client relationships the broker knows which client executed the order. This privileged information can be used by the broker to classify flow, i.e.,  toxic or benign,  and to devise strategies that mitigate adverse selection costs.

In the literature, models generally classify traders as informed or uninformed; see e.g., \cite{theonlygameintown}, \cite{copelandgalai}, 
\cite{grossman1980impossibility}, \cite{amihud1980dealership}, \cite{kyle1989informed},  \cite{kyle1985continuous},   and \cite{glosten1985bid}. In equity markets, many studies focus on informed flow (i.e., asymmetry of information) across various traded stocks, see e.g., \cite{easley1996liquidity} who study the probability of informed trading at the stock level, while our study focuses on each trade because we have trader identification. In FX markets, \cite{butz2019internalisation} develop a model for internalisation, and \cite{RoelAggregator} studies execution in an FX aggregator and the market impact of internalisation-externalisation strategies. 
Overall,  studies of toxic flow and information asymmetry do not make predictions of toxicity at the trade level. To the best of our knowledge, ours is the first paper in the literature to use FX data with trader identification  to predict the toxicity of each trade.

In our work, a trade is toxic if a client can unwind the trade within a given time window and make a profit (i.e., a loss for the broker). Toxic trades are not necessarily informed, nor informed trades are necessarily toxic. An uninformed client can execute a trade that becomes toxic for the broker because of the random fluctuations of exchange rates. Ultimately, the broker's objective is to avoid holding loss-leading trades in her books, so it is more effective to focus on market features and on each trade the broker fills, rather than on whether a particular client is classified as informed or uninformed. For simplicity, theoretical models in the literature assume traders are informed or uninformed, while in practice not all trades sent by one particular client are motivated by superior information.

The main contributions of our paper are as follows. We predict the toxicity of each incoming trade with machine learning and statistical methods,
such as logistic regression, random forests, a recursively updated maximum-likelihood estimator,
and a neural network (NNet).
We devise a novel algorithm to update the parameters of the NNet sequentially;
we call this rule of learning {\ourmodel},
which stands for projection-based unification of last-layer and subspace estimation.  We deploy our toxicity prediction models in a proprietary dataset, and we find that using a single model for all clients (employing client-specific features) outperforms the use of one model per client. We also find that, compared with the benchmarks,
the methodology we put forward 
attains the highest PnL and avoids the most losses by externalising toxic trades.

Our new method employs a NNet to compute the probability that a trade will be toxic.
After the outcome of each trade, toxic or benign, {\ourmodel} 
updates the parameters of the NNet.  To update the parameters efficiently at each timestep, {\ourmodel} follows three steps:
one,  split the last layer from the feature-transformation layers of a NNet;
two,  project the parameters of the feature-transformation layers onto an affine subspace; and
three,  devise a recursive formula to estimate a posterior distribution over the projected feature-transformation parameters and last-layer parameters.
Specifically, we extend the subspace NNet model (subspace NNets) of \cite{duran-martin22-subspace-ekf}
to classification tasks.
We also use the exponential-family extended Kalman filter (expfam EKF)
method of \cite{ollivier17-ekf-natural-gradient} and we follow the ideas of the
recursive variational Gaussian approximation (R-VGA) results of \cite{lambert2021-rvga} to obtain the update equations in {\ourmodel}.
Finally, we impose a prior independence between the hidden layers of the NNet and the
output layer, extending the work in the last-layer Bayesian NNet (last-layer BNNs).
Figure \ref{fig:model relationship} shows the relationship of {\ourmodel} to previous methods.
In short, {\ourmodel} is a statistically-efficient update rule to learn the parameters of a NNet sequentially.

\begin{figure}[H]
    \centering
    \begin{tikzcd}
    & \substack{\text{R-VGA} \\ \text{\cite{lambert2021-rvga}} } \arrow[d] & \\
    & \substack{\text{expfam EKF} \\ \text{\cite{ollivier17-ekf-natural-gradient}} } \arrow[d] &  \\
    \substack{\text{subspace NNets} \\ \text{\cite{duran-martin22-subspace-ekf}} } \arrow[r] & \text{\ourmodel}  & \substack{\text{last-layer BNN} \\ \text{\cite[S. 17.3.5]{pml2Book}} } \arrow[l]
    \end{tikzcd}
    \caption{Relationship of {\ourmodel} to other models.}
    \label{fig:model relationship}
\end{figure}

To evaluate the predictive performance of our model and the efficacy of the broker's strategy,
we use a proprietary dataset of FX transactions from 28 June 2022 to 21 October 2022.
Initially, the models are trained with data between 28 June and 31 July,
and  the remainder of the data (1 August to 21 October) is used to deploy the strategy,
i.e., use predictions of toxicity for each trade to inform the internalisation-externalisation strategy we develop.
During the deploy phase,
the maximum-likelihood estimator is continually updated with the running average of the toxic trades
and the parameters of the NNet are updated with {\ourmodel},
 while the models based on logistic regression and random forests are not updated.\footnote{Robustness checks in the appendix study the performance of the logistic regression and random forests  when the  models are re-trained weekly. }
For a given {\toxichorizon} and a {\cutoffprobability}, the strategy internalises the trade if the probability that the trade is toxic is less than or equal to the {\cutoffprobability}, otherwise it externalises the trade. 
We compute the PnL of all trades that the broker internalised and the losses she avoids by externalising trades. 
We find that
a NNet trained with {\ourmodel}
delivers the best combination of PnL  and avoided loss across all {\toxichorizon}s we consider in this paper.

Finally, we find that a universal model is more advantageous than one model per trader.
That is, we obtain higher accuracies (when predicting toxic trades) when we train one model for all traders than when we train one model per trader;
higher accuracies result from having more data.
When one restricts to one model per trader,
the model for traders with fewer transactions is outperformed by a universal model that is trained on more datapoints.
We also find that if we build a universal model that does not consider
the inventory, cash, and recent activity of clients (i.e.,  the broker does not use the identification of the trader), the performance of {\ourmodel} deteriorates substantially when compared to the performance of a model that includes the identity and unique features of the trader.
Thus, in our dataset, client-specific variables does add value to predict the  toxicity of trades.

The remainder of paper is organised as follows.
Section \ref{sec:data-analysis} describes the data.
Section \ref{sec: toxicity} defines toxicity, provides statistics about the clients in the dataset, and illustrates the toxicity profiles of clients.
Section \ref{sec:our-method} introduces \ourmodel, which is a fast and statistically-efficient Bayesian procedure for online training of neural networks.
Section \ref{sec:deployment-methods} shows implementation details.
Section \ref{section:traders-real} uses proprietary datasets to evaluate PULSE against alternative methods.
Section \ref{sec:conclusions} presents conclusions.
We collect proofs, together with additional robustness checks, in the appendix.

\section{Data and preliminary analysis}
\label{sec:data-analysis}
We employ data for the currency pair EUR/USD from LMAX Broker and from LMAX Exchange
for the period 28 June 2022 to 21 October 2022.\footnote{\url{www.lmax.com}.}
For each liquidity taking trade filled by the broker, we use the direction of trade (buy or sell), 
the timestamp when LMAX Broker processed the trade, and the volume of the trade. 
Also, we use the best quotes and volumes available in LMAX Exchange at a microsecond frequency. In contrast to LMAX Broker, traders who interact in the limit order book (LOB) of LMAX Exchange do not know the identity of their counterparties. The LOB uses 
price-time priority to clear supply and demand of liquidity --- as in traditional electronic order books in equity markets, such as those of Nasdaq,
Euronext, and the London Stock Exchange.

Table \ref{tab:datset prop} shows summary statistics for the trading activity of six clients
of LMAX Broker in the pair EUR/USD.

\begin{table}[H]
\centering
\begin{tabular}{l|rrrr}
\noalign{\vskip 1mm}
  \hline\hline
  \noalign{\vskip 1mm}
 & Number of trades & Total volume & Avg daily volume \\
 & & \multicolumn{2}{c}{in \euro 100,000,000 }\\
\noalign{\vskip1mm}
 \hline
 \noalign{\vskip-1mm}
 \noalign{\vskip 2mm}
Client 1 & 312,073 & 43.702 & 0.520 \\
Client 2 & 56,705 & 3.006 & 0.036 \\
Client 3 & 28,185 & 3.278 & 0.039 \\
Client 4 & 27,743 & 0.456 & 0.005 \\
Client 5 & 23,938 & 27.483 & 0.348 \\
Client 6 & 13,379 & 5.379 & 0.064 \\
\hline
Total & 462,023 & 83.304 & 1.012 \\
\noalign{\vskip 1mm}
  \hline
  \hline
\end{tabular}
\caption{Trading activity in  the pair EUR/USD between clients and LMAX Broker over the period 28 June 2022 to 21 October 2022 .
Volumes are reported in one hundred  million euros.
}\label{tab:datset prop}
\end{table}

Below, we work with the data of Clients 1 to 6  and we assume that the broker quotes her clients the best available rates in LMAX Exchange net of fees.
Transaction costs in FX are around \$3 per million euros traded
(see e.g.,  \cite{cartea2023optimal}),
so this assumption  provides clients with a discount of 
\$3 per million euros traded  when trading with the broker.\footnote{
In practice, not all clients receive the same bid and ask quotes.
}

\section{Toxicity}\label{sec: toxicity}

In this paper, a trade is toxic over a given time window if the client can unwind the trade at a profit within the time window. Instead of classifying traders as informed or uninformed, the broker assesses the probability that each trade becomes toxic within a specified time window. Not all trades sent by better informed clients will be toxic, and not all trades sent by less informed clients will be benign. 
Thus, our models aim to predict price movements based on current features regardless of whether the trader is informed or not. 
Our  methods, however, include the identity of the trader, so predicting toxicity of a trade will depend, among other features, on how often the client executed toxic trades in the past.

Denote time by $t\in\mfT=[0,T]$, where $0$ is the start of the trading day and
$T$ is the end of the trading day.
From this point forward, we use `exchange rate' and `prices' interchangeably. 
The best ask price  and best bid price in the LOB of LMAX Exchange are denoted by
$(S^a_t)_{t\in\mfT}$ and $(S^b_t)_{t\in\mfT}$, respectively.
Let $\valuetoxichorizon$ be a {\toxichorizon} such that $0<\valuetoxichorizon\ll T$,  and let $t\in[0,T-\valuetoxichorizon]$.
We define the two stopping times
\begin{equation*}
    \tau^+_t = \inf\left\{u\in [t,T]: S^b_u > S^a_t\right\}  \qquad\text{and}\qquad 
    \tau^-_t = \inf\left\{u\in [t,T]: S^b_t > S^a_u\right\}  \,,
\end{equation*}
 with the convention that $\inf\,\emptyset = \infty$.
The stopping time $\tau^+_t$ is the first time after $t$ that the best bid price is above the best ask price
at time $t$.
If $\tau^+_t<\infty$, a buy trade executed at $S_t^a$ becomes profitable for the client  at time $\tau^+_t$ 
before the end of the trading day because the client can unwind her position and collect the profit
\begin{equation}\label{eqn: profit 1}
S_{\tau^+_t}^b - S_t^a>0\,.    
\end{equation}
Similarly, $\tau^-_t$ is the first time after $t$ that the best ask price is below the best bid price at time $t$.
If $\tau^-_t<\infty$, a sell trade executed at $S_t^b$ is profitable  for the client at time $\tau^-_t$
before the end of the trading day because the client can unwind her position and collect the profit 
\begin{equation}\label{eqn: profit 2}
S_t^b - S_{\tau^-_t}^a >0  \,.    
\end{equation}

\begin{definition}[Toxic trade]\label{def: toxicity} 
Let  $\valuetoxichorizon > 0$ be a {\toxichorizon}.
A client's buy (resp.~sell) filled by the broker at time $t$ is toxic for the broker if $\tau^+_t \leq t+\valuetoxichorizon$ (resp.~if $\tau^-_t \leq t+\valuetoxichorizon$). 
\end{definition}

The above definition captures the broker's exposure to adverse selection.
A trade is labelled as toxic if over a given time window the client had the option to unwind the trade at a profit,
in which case it would be a loss-leading trade for the broker. However, one cannot verify, ultimately, if the potentially toxic trade materialised as a loss to the broker or to another market participant. We do not have enough information to track each step, or potential step, in the life cycle of a trade to determine who made a loss or a gain --- to make this assessment requires perfect knowledge of all trades by all market participants.\footnote{
There are alternative ways of defining toxic flow with our proprietary data set. For instance, the broker values a trade after unwinding it with the first trade in the opposite direction, i.e., a first-in-first-out inventory valuation. }

Figure \ref{fig:toxic-profile}  plots the trajectories of $S^a_t$ and $S^b_t$ for EUR/USD between 10:00:00 am 
and 10:00:10 am in LMAX Exchange on 28 June 2022.
The dotted line is the best ask price and the dash-dotted line is the best bid price.
The solid horizontal lines are the best ask price and the best bid price at 10:00:00 am.

\begin{figure}[H]
    \centering \includegraphics[width=0.8\linewidth]{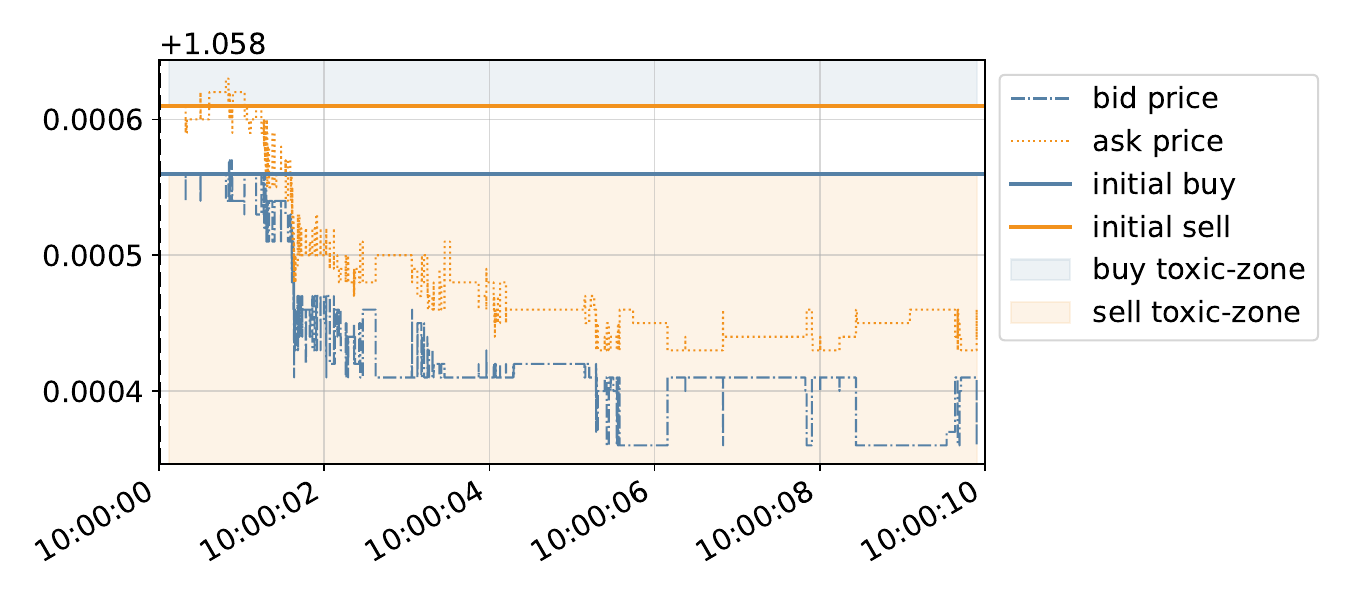}
    \caption{A client's sell trade that becomes toxic for the broker after a few seconds of filling the trade. The $x$-axis is time and the $y$-axis is in units of USD. }
    \label{fig:toxic-profile}
\end{figure}

In the figure, if a client buys from the broker at the best ask price at time $t=$10:00:00 am, then there is no opportunity for the trader to unwind the trade at a profit in the first ten seconds after the trade. 
However, had the trader sold to the broker at the best bid price at time $t=$10:00:00 am, then shortly after 10:00:01 the trade would be in-the-money for the client (i.e., toxic for the broker).

\subsection{Toxicity profiles}\label{sec: toxicity profiles}

For a {\toxichorizon} $\valuetoxichorizon>0$, we use both the client data and the LOB data to determine if the trades filled by LMAX Broker were toxic over the period $\valuetoxichorizon$. Table \ref{tab:toxicity proportion clients} shows the percentage of toxic trades
executed by each client
for $\valuetoxichorizon \in \{1,\, 5,\, 10,\,20,\, 30,\, 40,\, 50,\, 60,\, 70\}$ seconds. 

\begin{table}[H]
\centering
\begin{tabular}{l|rrrrrrrrr}
\noalign{\vskip 1mm}
  \hline\hline
  \noalign{\vskip 1mm}
& \multicolumn{9}{c}{{\toxichorizon} $\valuetoxichorizon $ in seconds}  \\
 & 1 & 5 & 10 & 20 & 30 & 40 & 50 & 60 & 70 \\
\noalign{\vskip1mm}
 \hline
 \noalign{\vskip-1mm}
 \noalign{\vskip 2mm}
 Client 1 & 6.7 & 25.7 & 38.4 & 51.5 & 58.7 & 63.4 & 66.7 & 69.3 & 71.3 \\
Client 2 & 7.0 & 28.6 & 42.4 & 56.1 & 63.0 & 67.5 & 70.7 & 73.1 & 74.9 \\
Client 3 & 7.0 & 26.0 & 38.6 & 51.1 & 58.2 & 62.6 & 65.8 & 68.2 & 70.4 \\
Client 4 & 3.4 & 18.5 & 30.7 & 44.3 & 52.1 & 56.8 & 60.6 & 63.5 & 66.0 \\
Client 5 & 8.3 & 22.1 & 31.7 & 42.9 & 50.1 & 55.4 & 59.6 & 62.2 & 63.9 \\
Client 6 & 5.9 & 26.4 & 40.2 & 53.3 & 60.9 & 65.7 & 69.0 & 71.8 & 73.9 \\
\hline
All clients & 6.6 & 25.5 & 38.2 & 51.2 & 58.4 & 63.1 & 66.5 & 69.0 & 71.0 \\
\noalign{\vskip 1mm}
  \hline
  \hline 
\end{tabular}
\caption{Proportion of toxic trades (in \%)  between 28 June 2022 and 21 October 2022.}\label{tab:toxicity proportion clients}
\end{table}

As expected, for short toxicity horizons only a small proportion of trades are toxic (e.g., 6.6\% for a one second horizon),
but as the toxicity horizon increases, the proportion of toxic trades grows considerably
(e.g., it is roughly 70\% after one minute).
A simple mathematical argument can help us justify what we observe in the data.
Consider a trader who sends a liquidity taking trade to the broker when the spread in the market is
$\mathfrak{s}$ and suppose that the profitability of unwinding the trade,
which is $-\mathfrak{s}$ at time zero,
diffuses according to a scaled Brownian motion $\sigma W_t$ with $\sigma>0$.
From the reflection principle,
the probability that such a trade becomes toxic at any point between zero and $\valuetoxichorizon$ seconds
is given by
\begin{equation}
\mathbb{P}\bigg(\sup_{t\in[0,\valuetoxichorizon]} \sigma\,W_t \geq \mathfrak{s}\bigg) = \mathbb{P}\bigg(\sup_{t\in[0,\valuetoxichorizon]} W_t \geq \frac{\mathfrak{s}}{\sigma}\bigg) = 2\,\mathbb{P}\bigg(W_\valuetoxichorizon \geq \frac{\mathfrak{s}}{\sigma}\bigg) = 2\,\bigg(1- \Phi\Big(\frac{\mathfrak{s}}{\sigma\sqrt{\valuetoxichorizon}}\Big)\bigg)\,,
\end{equation}
where $\Phi$ is the standard normal cumulative distribution function.
As the horizon $\valuetoxichorizon\to\infty$,
the probability that the trade is toxic at some point converges to one.
This does not mean that the broker loses money on each trade.
That would happen only if the broker failed to hedge in the lit market and every liquidity unwound exited at the first profitable moment.
In reality, the broker’s inventory shifts constantly due to both incoming trades and her own hedging in the lit market.

Arguably, a client can be labelled as toxic according to the percentage of their trades that are in-the-money after a set time frame, e.g., after $\valuetoxichorizon$ in Definition \ref{def: toxicity} above.
Figure \ref{fig:sharpness-profile} shows the toxicity profiles of two clients trading EUR/USD on 8 July 2022 with LMAX Broker. The figures summarise all trades by clients A and B as follows. For each trade on 8 July 2022, the plot shows the profitability, from the client's perspective, of unwinding each trade a given number of seconds after the trade. Here, the $x$-axis goes from zero to ten seconds and the $y$-axis is in dollars per million euros traded.

\begin{figure}[htb]
    \centering
    \includegraphics[width=0.45\linewidth]{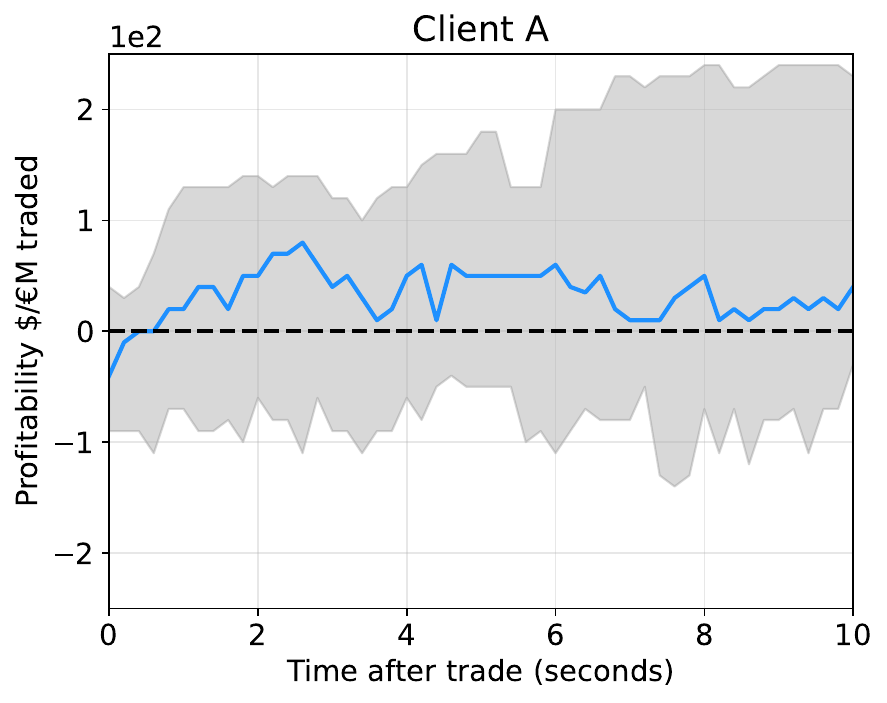}
    \includegraphics[width=0.45\linewidth]{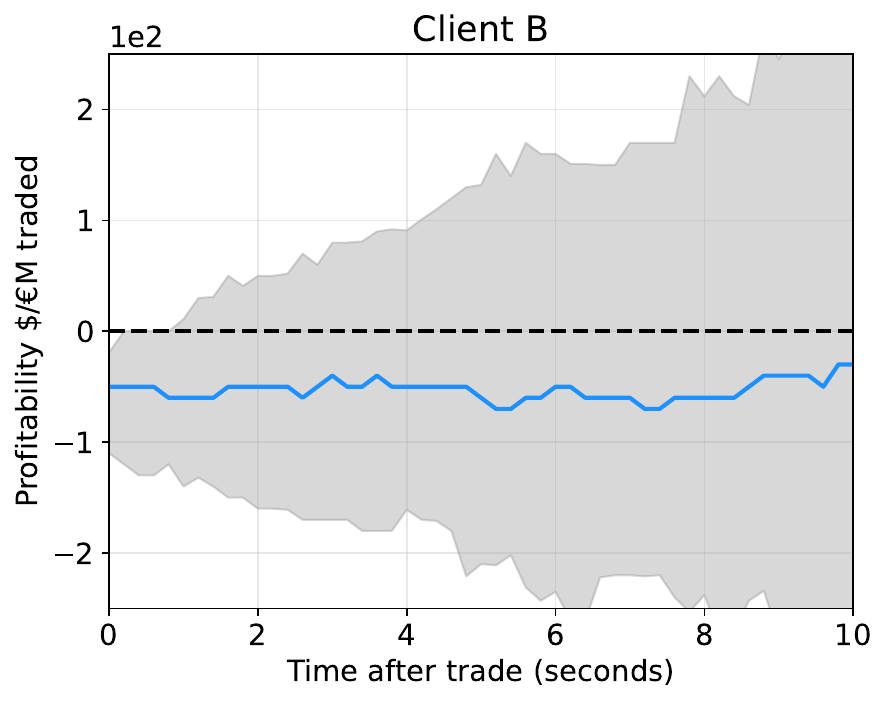}
    \caption{Profitability in dollars per million euros traded after a trade is executed. Panels correspond to two different clients. Blue line is the median trajectory and  grey region is the 90\% trajectory region. The $x$-axis is time and the $y$-axis is the profitability from the point of view of the client. }
    \label{fig:sharpness-profile}
\end{figure}

From Figure \ref{fig:sharpness-profile}, and all else being equal, a broker would prefer to provide liquidity to client B
instead of client A.
More than 50\% of the times that client A trades, the broker is exposed to making a loss on the trade in less than half a second.
On the other hand, the median trajectory of profitability (blue line) for client B is below zero.\footnote{Clearly, if the broker uses the flow from client A to inform other decisions and other investment strategies that are profitable for the broker, the net effect of trading at a loss with client A could be positive for the broker; see \cite{cartea2022brokers}.}

\subsection{Features to predict toxic trades}\label{sec:features}
For each client, the broker uses features that reflect (i) the state of the LOB, (ii) recent activity in the LOB,
and (iii) the cash and inventory of the client in the EUR/USD currency pair.\footnote{We use the transactions with the broker to compute the cash and inventory of the clients. Clients can trade elsewhere but this is unknown to the broker.}
Our features build on \cite{ait2022and}, who propose using three different clocks to aggregate LOB data. 
Here, we compute eight LOB statistics,
each with three clocks (time, volume, and transaction),
and seven backward-looking \textit{clock} intervals of increasing length;
thus, we have $8 \times 3 \times 7 = 168$ clock-based features.
We also employ fifteen other features (e.g., cash and inventory),
so there are 183 features per client.

For each clock and a given interval in the past (as measured by the clock), we employ the following eight features: 
(a) volatility of the midprice, 
(b) number of trades executed by the client in the interval, 
(c) number of updates in the best quotes of LMAX Exchange in the interval, 
(d) return of the midprice over the interval, 
(e) average transformed volume in the best bid price, 
(f) average transformed volume in the best ask price, 
(g) average spread, and 
(h) average imbalance of the best available volumes, 
where ``average'' is simple arithmetic mean within an interval.

The three clocks provide alternative ways of deciding the datapoints that fall in  a time  interval. 
For example, 
with a time clock, the spread over the last two seconds is computed with all datapoints in the last  2 seconds. 
With a volume clock,
the spread over the last $V$ units of volume traded is computing
with all the data in the past (chronologically) until we gather $V$ units of volume traded. 
With a transaction clock, the spread over the last $k$ transactions is computed with the datapoints from the last $k$ transactions.
See \ref{app: features} for more details.

For a given client $c\in\mcA$, we use the following additional features:
(i)  cash,
(ii)  inventory,
(iii) volume of the order,
(iv)  spread  in the market just before the order arrives, 
(v)  imbalance in the LOB just before the order arrives, 
(vi)  transformed available volume in the best bid,
(vii)  transformed available volume in the best ask,
(viii) last ask price at the time of trade,
(ix) last bid price at the time of trade,
(x) last midprice at the time of trade,
(xi) total number of market updates since starting date,
(xii) number of trades made by client $c$,
(xiii) total number of trades executed by all clients,
(xiv) volatility estimate of the mid-price, and
(xv) proportion of previous toxic trades executed by client $c$.
{
We apply log-transformations to stabilise scale and limit outlier leverage:
a signed log-transformation of the form $\text{sign}(x)\,\log(1+|x|)$ for (i) cash and (ii) inventory, and a $\log(1+x)$ transform for volumes in (iii), (vi), and (vii). 
These monotone transforms reduce heteroskedasticity and outlier influence and improve numerical conditioning for the downstream models;
see e.g., \cite{west2022best}.
}
The fifteen features above account for both the state of the LOB, 
and the cash and inventory of the client.
The remaining 168 features account for recent activity in the LOB.
These are features (a)--(h) above, measured for each of the seven intervals and for each of the three clocks to obtain a total of $8\times 7\times 3 = 168$ features.
Thus, for each client, we employ $15 + 8\times 7 \times 3 = 183$ features to predict the toxicity of their trades. 

{
In  \ref{sec:independent-clock-performance}
we illustrate that the performance of \ourmodel does not improve with
additional information from the volume and transaction clock.
In what follows we use the features provided by the three clocks to align with prior work and for completeness.
}

\section{The {\ourmodel} method}
\label{sec:our-method}

{
For a trade from client $c\in\mcA$ filled at time $t$, let $\xt\in\reals^M$ denote the features observed at $t$ (with $M$ the number of features).
Define $\yt\in\{0,1\}$ as the toxicity indicator evaluated $\valuetoxichorizon>0$ after $t$, with
$\yt=1$ if the trade is toxic within the horizon, and $0$ otherwise.
Thus, $\xt$ is observed at time $t$, whereas $\yt$ is known at time $t+\valuetoxichorizon$.

For each trading day we collect observations $(\data_{t_i})_{i\in I}$,
where $\data_{t_i}=(\xt[t_i],\yt[t_i])$, $I=\{1,2,\dots,N\}$, $t_i$ is the arrival time of observation $i$,
and $N$ is the number of trades that day.
For notational convenience we write $\data_i$ for $\data_{t_i}$ and refer to the dataset as $(\data_i)_{i\in I}$.
For $n\in\N$, let $\data_{1:n}=\{\data_1,\ldots,\data_n\}$ denote the first $n$ observations;
in particular, for $n\le N$ we have $\data_{1:n}\subseteq \data_{1:N}$.
}

{
Conditioned on the features $\xt$, the toxicity indicator $\yt \in \{0,1\}$ is modelled as Bernoulli random variable
with probability mass function
}
\begin{equation}\label{eq: model for yt}
    p(y \,\vert\, \bm\theta; \xt)
    = \text{Bern}\Big(y \,\vert\, \sigma(\plast^\intercal\, g(\phidden; \xt))\Big),\qquad y\in\{0,1\}\,,
\end{equation}
where $g: \reals^M \to \reals^L$ is the output-layer of a NNet,
{
${\bm\theta} = (\plast, \phidden)$ are the parameters of the neural network,
$\plast$ are the parameters of the \textit{last layer},
and
$\phidden$ are the parameters in the \textit{hidden layers}.
}
Here, and throughout the paper, we adopt the convention that $p$ denotes a likelihood function 
or a posterior density function. 
The function $\text{Bern}(\cdot \,\vert\, \cdot)$ is given by
\begin{equation*}
    \text{Bern}(a\,\vert\,b) = b^a\,(1-b)^{1-a},\ a\in\{0, 1\} \text{ and } b\in(0, 1)\,.
\end{equation*}
We refer to $\plast \in \reals^L$ as the last-layer parameters
and we refer to $\phidden \in \reals^D$ as the feature-transform parameters.
The function $\sigma(\plast^\intercal\, g( \phidden;\xt))$ is a NNet for classification
where $\sigma(x) = (1 + \exp(-x))^{-1}$ is the  sigmoid function.
Figure \ref{fig:pulse-description} shows a graphical representation of
the parameters that {\ourmodel} updates when the
NNet is a multilayered-perceptron (MLP).
Although we choose an MLP in the experiments, {\ourmodel} can
be used with any  NNet architecture.

\begin{figure}[H]
\centering
\scalebox{0.6}{

\begin{tikzpicture}[scale=1]
	\tikzstyle{short} = [ shorten >=20, shorten <=20 ]
	\draw[thick] (0,-3)
		circle (5pt)
		node[below, yshift=-10] {$p(1 | \sigma({\bf w}^\intercal g(\boldsymbol{{\bf x}_t;\psi}))$}
	;
	
	\foreach[count=\i] \x in {-2, -1, 1, 2}{
		\draw[thick, short]
			 (\x, -1) -- (0, -3);
		\draw[thick] (\x, -1) circle (5pt);
	};
	
	\foreach[count=\i] \x in {-3, -2, 2, 3}{
		\draw[thick] (\x, 1) circle (5pt);
	};
	
	\foreach[count=\i] \x in {-2, 2}{
		\draw[thick] (\x, 3) circle (5pt);
	};
	
	\foreach \i in {-3, -2, 2, 3}{
		\foreach \j in {-2, 2} {
			\draw[thick, short]
			(\i, 1) -- (\j, 3);
		};
	};
	
	\draw[thick, short] (0, 5) -- (-2, 3);	
	\draw[thick, short] (0, 5) -- ( 2, 3);	
	\filldraw[thick] (0, 5) circle (5pt)
	node[above, yshift=10] {${\bf x}_t\in\mathbb{R}^M$};
	
	\draw (0, -1) node {\huge $\ldots$};
	\draw (0,  0) node {\huge $\vdots$};
	\draw (0,  1) node {\huge $\ldots$};
	\draw (0,  3) node {\huge $\ldots$};

	\draw
	[thick, decorate,decoration={brace,amplitude=10pt}]
	(-4,-0.6) -- (-4,5) node [left, midway, xshift=-10pt]
	{\large ${\boldsymbol{\psi}}$}; 
	
	\draw
	[thick, decorate,decoration={brace,amplitude=10pt}]
	(-4,-3) -- (-4,-0.9) node [left, midway, xshift=-10pt] {\large ${\bf {w}}$};
	
	\draw
	[thick, decorate,decoration={brace, amplitude=10pt, mirror}]
	(4,-3) -- (4,5) node [right, midway, xshift=18] {\large ${\boldsymbol{\theta}} = (\plast, \phidden)$};

\end{tikzpicture}}
\caption{
    {\ourmodel} architecture for an MLP. The MLP is parameterised
    by $\bm\theta = (\phidden, \plast)$, where
    $\phidden$ are the parameters in the hidden layers
    and $\plast$ are the parameters in the last layer.
}
\label{fig:pulse-description}
\end{figure}

\subsection{Sequential learning}
{
We update the model parameters sequentially after each observed toxicity label to incorporate new information.
Specifically, we update
}
${\bm\theta} = (\plast, \phidden)$ in \eqref{eq: model for yt} after each new $y_t$ is observed,
i.e., after observing if the trade is toxic.
In practice, the dimension  $D$ of the featured-transformed parameters and the dimension $M$ of the feature space satisfy $D\gg M$,
so  it is costly to update $\bm\theta$
after each new observation using standard training techniques.
Thus, the literature proposes various approaches to estimate the parameters
of a NNet at a lower computational cost.
In this paper, we build on two of these methods: lottery-ticket and last-layer methods.

Lottery-ticket methods exploit the over-parametrisation of NNets,
in the sense that
``a randomly-initialised dense NNet contains subnetworks that, when trained in isolation,
reach test accuracy comparable to that of the original network'', \cite{frankle2019lottery}.
The lottery-ticket hypothesis states that such a subnetwork exists, and
subnetworks satisfying the lottery-ticket hypothesis are called \textit{winning tickets}.
\textit{Winning tickets} are linear projections of the NNet parameters onto a subspace; see  \cite{li2018subspacenn} and \cite{dof-dnns-2021}.
\cite{duran-martin22-subspace-ekf} use the lottery-ticket hypothesis with a
relatively small linear subspace, and 
use the extended Kalman filter (EKF) algorithm to propose a sequential update of the subspace parameters.

Alternatively, last-layer methods pre-train the NNet parameters in a \textit{warmup} phase
and then perform sequential updates on the last-layer parameters,
see \citet[S. 17.3.5]{pml2Book}.
Here, we employ both methods.
Specifically, we propose a Bayesian approach to  update the parameters of a NNet sequentially
for both the subspace parameters in the hidden-layers
and all of the parameters in the last layer.
While previous literature focuses on  updating either the last-layer or all parameters when performing online learning,
ours is the first work that projects the parameters of the hidden layer and
updates all of the units in the last layer.
{
This decomposition enables full-rank updates of the last-layer parameters
while restricting hidden-layer updates to a linear subspace.
In doing so, it balances the rapid sequential updates typical of subspace neural networks
with the statistical efficiency characteristic of last-layer methods.
}

{
In particular, we modify
}
\eqref{eq: model for yt} and
decompose
{
the parameters of the hidden layer
}
$\phidden\in\reals^D$ as an affine projection of the form
\begin{equation}\label{eq:subspace-decomposition}
    \phidden = {\bf A}\,{\bf z} + {\bf b}\,,
\end{equation}
where ${\bf A}\in\reals^{D\times d}$ is the fixed projection matrix  and
${\bf z}\in\reals^d$ are the projected (subspace) parameters such that $d \ll D$, 
and ${\bf b}\in\reals^D$ is the offset term.
{
The decomposition \eqref{eq:subspace-decomposition}
provides a linear-subspace formulation of the lottery-ticket hypothesis
and enables efficient updates of the hidden layer parameters.
}
With this projection, we rewrite \eqref{eq: model for yt} as 
$p(y \,\vert\, \phiddensub, \plast; \xt) =
\text{Bern}\Big(y \,\vert\, \sigma(\plast^\intercal\, g({\bf A}\,\phiddensub + {\bf b};\, \xt))\Big)$,
$y\in\{0,1\}$.
To simplify notation, we define $h(\phiddensub; \xt) = g({\bf A}\,\phiddensub + {\bf b}; \xt)$ and write
\begin{equation}\label{eq:model-yt-sub}
    p(y \,\vert\, \phiddensub, \plast; \xt)
    = \text{Bern}\Big(y \,\vert\, \sigma(\plast^\intercal\, h(\phiddensub; \xt))\Big).
\end{equation}

{
Our procedure has two stages:
(i) an offline warmup phase that estimates ${\bf A}$ and ${\bf b}$ and uses  ${\cal D}_\text{warmup}$ to select hyperparameters
and
(ii) an online phase that performs fast sequential updates of $(\plast,\phiddensub)$ on the live stream ${\cal D}_\text{deploy}$ while holding ${\bf A},{\bf b}$ fixed.
Here, ${\cal D}_\text{warmup}$ is a one-off historical window,
whereas ${\cal D}_\text{deploy}$ grows over time and is used for online learning and evaluation.
}

{
Figure~\ref{fig:diag-warm-up-deploy} illustrates the workflow: an offline warmup phase using ${\cal D}_\text{warmup}$ (Section~\ref{subsec:warm-up}) followed by an online deploy phase using ${\cal D}_\text{deploy}$ (Section~\ref{subsec:deploy}). The deploy phase proceeds indefinitely as trades arrive (the time index $T$ may grow over time).
}
\begin{figure}[h!]
    \centering
    \begin{tikzpicture}
\draw[ultra thick, ->] (0,0) -- (12cm,0);

\foreach \x in {0, 2,4,6,8,10,12} \draw (\x cm,3pt) -- (\x cm,-3pt);

\draw (0, 0) node[below=3pt] {$t_0$};
\draw (4, 0) node[below=3pt] {$t_\text{warmup}$};
\draw (12, 0) node[below=3pt] {$T$};

\draw[dashed] (4, 0) -- (4, 1.5) node[above=3pt, align=center]
{\scriptsize Initialise \\ \scriptsize$\phi_0(\plast)$, $\varphi_0(\phiddensub)$};

\draw[
black, thick,
decorate, decoration={brace, amplitude=5pt}
]
(0, 0.5) -- (3.9, 0.5) node[midway, above=4pt, align=center] {\scriptsize warmup stage\\ \scriptsize estimate $\bf A$, $\bf b$};

\draw[
black, thick,
decorate, decoration={brace, amplitude=5pt},
]
(12, -0.7) -- (4.1, -0.7) node[midway, below=4pt, align=center] {\scriptsize deploy stage\\ \scriptsize estimate $\vdlast, \vdhidden\,\; \forall t$ };

\end{tikzpicture}
    \caption{
    Warmup and deployment stages. We use all data available from
    $t_0$ to $t_\text{warmup}$ to estimate $\bf A$ and $\bf b$.
    At $t_\text{warmup}$, we initialise the variational approximations
    $\phi_0(\plast)$ and $\varphi_0(\phiddensub)$. Finally,
    for  $t > t_\text{warmup}$, we estimate $\plast_t$ and $\phiddensub_t$.
    }
    \label{fig:diag-warm-up-deploy}
\end{figure}

{

}

\subsubsection{Warmup phase: estimating the projection matrix and the offset term}\label{subsec:warm-up}
{
This phase estimates  ${\bf A}$ and ${\bf b}$, and assigns prior distributions for $\plast$ and $\bf z$.
}
Given the size of the dataset, we divide $\dwarmup$ into $B$ 
non-intersecting random batches
$\data_{(1)}, \ldots, \data_{(B)}$ such that $$\bigcup_{b=1}^B \data_{(b)} = \dwarmup\,.$$

To estimate ${\bf b}$ and $\bf A$, we use mini-batch stochastic gradient descent (SGD) over $\dwarmup$ to  minimise the negative loss-function
\begin{equation}\label{eq:neg-log-likelihood}
    -\log p(\data \,\vert\, {\bm\theta}) = -\sum_{n=1}^N \log p(\yt[n]\,\vert\,{\bm\theta}, \xt[n])\,,
\end{equation}
where $\data$ is any random batch.
The vector ${\bf b}$ is given by
\begin{equation}
    {\bf b} = \argmin_{\bm\theta}-\log p(\data\,\vert\,\bm\theta).
\end{equation}

Singular-value decomposition (SVD) over the iterates of the SGD optimisation procedure gives the projection matrix ${\bf A}$.
To avoid redundancy, we skip the first $n$ iterations and store the iterates every $k$ steps. The dimension of the subspace
$d$ is found via hypeparameter tuning and 
convergence to a local minimum of \eqref{eq:neg-log-likelihood} is obtained through multiple passes of the data.
Algorithm \ref{algo:MAP-SGD} shows the training procedure for a number  $E$ of epochs.
\vspace{0.3cm}

\begin{algorithm}[H]
def \textbf{warmupParameters}: \\
  Initialise model parameters ${\bm\theta} = (\phidden, \plast)$\\
  \ForEach{epoch $e=1,\ldots,E$}{
    \ForEach{batch $m=1,\ldots,M$}{
      $\text{gradient} = -\nabla_{\bm\theta}\log p(\data_{(m)}
      \,\vert\, {\bm\theta})$\\
      ${\bm\theta}^{(e)} \gets {\bm\theta}^{(e-1)} - \alpha\, \bm\kappa(\text{gradient})$
    }
}
\caption{MAP parameter estimation via batch SGD}
\label{algo:MAP-SGD}
\end{algorithm}

\vspace{0.2cm}

In Algorithm \ref{algo:MAP-SGD}, the parameter $\alpha$ is the learning rate, and
the function $\bm\kappa: \reals^M \to \reals^M$ is the per-step transformation of the Adam algorithm;
see \cite{kingma14-adam}.
At the end of the $E$ epochs, we obtain ${\bm\theta}^{(E)} = (\phidden^{(E)}, \plast^{(E)})$.
Then, the offset term ${\bf b}$ is given by
\begin{equation*}
    {\bf b} = \phidden^{(E)},
\end{equation*}
and we stack the history of the SGD iterates.
To avoid redundancy, we skip the first $n$ iterates of the SGD, which are stored at every $k$ steps. We let
\begin{equation*}
{\cal E} = \begin{bmatrix}
    \horzbar & \phidden^{(n)} & \horzbar\\
    \horzbar & \phidden^{(n + k)} & \horzbar\\
    \horzbar & \phidden^{(n + 2k)} &\horzbar\\
    & \vdots & \\
    \horzbar & \phidden^{(E)} & \horzbar\\
\end{bmatrix} \in \reals^{\hat{E}\times D}\,,
\end{equation*}
where $\hat{E} = E - n +1$.
With the SVD decomposition ${\cal E} = {\bf U}\,\bm\Sigma\,{\bf V}$ and the first $d$ columns of the matrix ${\bf V}$,
the projection matrix is 
\begin{equation*}
    {\bf A} = \begin{bmatrix}
    \vertbar & \vertbar &  & \vertbar \\
    {\bf V}_1 & {\bf V}_2 & \ldots & {\bf V}_d \\
    \vertbar & \vertbar & & \vertbar \\
    \end{bmatrix},
\end{equation*}
where ${\bf V}_k$ denotes the $k$-th column of ${\bf V}$. 

\subsubsection{
Deploy phase: online estimation of last-layer and subspace-feature-transform parameters
}\label{subsec:deploy}

Here, we derive a novel, sample-free, and closed-form update rule that estimates the parameters of
\eqref{eq:model-yt-sub} sequentially.
Specifically, in Proposition \ref{app: prop: fixed point} in the Appendix, we find a set of fixed-point equations for the update rule.
These equations need many iterations to converge  and are computationally inefficient.
Thus, Corollary \ref{cor:rvga-order-2} computes the gradient with respect to the subspace parameters to simplify the computations. 
Finally, Theorem  \ref{theorem:subspace-last-rvga} uses a Taylor expansion of the measurement model
to obtain a closed-form solution to the set of fixed-point equations. This is computationally efficient because the update can be obtained in a single iteration.

We introduce Gaussian priors for both $\plast$ and ${\bf z}$ at the beginning of the deploy stage.
Let $n=0$ denote the last timestamp in the warmup dataset and $n=1$ the first timestamp of the deploy dataset.
Denote by $\phi_n$ and $\varphi_n$ the posterior distribution estimates for $\plast$ and ${\bf z}$ at time $t$, respectively.
The initial estimates are given by 
\begin{align*}
    \phi_0(\plast) &= \normdist{\plast}{\plast^{(M)}}{\sigma^2_\plast\,{\bf I}\,},\\
    \varphi_0({\bf z}) &= \normdist{{\bf z}}{{\phidden}^{(M)}\,{\bf A}}{\sigma^2_\phiddensub\,{\bf I}\,},
\end{align*}
where $(\plast^{(M)}, \phidden^{(M)})$ are the last iterates in the warmup stage,
$\sigma^2_\plast$ and $\sigma^2_\phiddensub$ are the coefficients of the prior covariance matrix,
${\bf I}$ is the identity matrix,
and recall that ${\bf A}$ is the projection matrix.

For  $n \geq 1$, the variational posterior estimates are given by
\begin{equation*}
\begin{aligned}
    \vdlast[n] &= \normdist{\plast}{\mlast_n}{\covlast_n}\quad\text{ and }\quad
    \vdhidden[n]  = \normdist{{\bf z}}{\mhidden_n}{\covhidden_n}. 
\end{aligned}
\end{equation*}
Next, to find the posterior parameters $\mhidden_n, \mlast_n, \covhidden_n, \covlast_n$, we recursively solve
the following variational inference (VI) optimisation problem
\begin{equation}\label{eq:subspace-last-rvga}
    \mhidden_n,\, \mlast_n,\, \covhidden_n,\, \covlast_n =
    \argmin_{\mhidden,\, \mlast,\, \covhidden,\, \covlast}
    \text{KL}\left(
    \normdist{\plast}{\mlast}{\covlast}\,\normdist{{\bf z}}{\mhidden}{\covhidden} || \phi_{n-1}(\plast)\,\varphi_{n-1}({\bf z}) \,p(\yt[n] \,\vert\, {\bf z}, \plast; \xt[n])
    \right)\,,
\end{equation}
where KL is the Kullback--Leibler divergence
\begin{equation*}
    \text{KL}(p(x) || q(x)) = \int p(x)\log\left(\frac{p(x)}{q(x)}\right) \diff x\,,
\end{equation*}
for probability density functions  $p$ and $q$  with the same support.
The optimisation in \eqref{eq:subspace-last-rvga} generalises the update rule for the Kalman filter when the parameters do no\textbf{}t have a drift; see e.g., \cite{lambert2021-rvga}.
The following theorem shows the update and prediction equations of the {\ourmodel} method.

\begin{theorem}[\ourmodel]\label{theorem:subspace-last-rvga}
      Suppose $\log p(\yt[n] \,\vert\, \phiddensub, \plast; \xt[n])$ is differentiable with respect
    to $(\phiddensub, \plast)$ and the observations $\{\yt[n]\}_{n=1}^N$ are conditionally independent over
    $(\phiddensub, \plast)$.
    Write the mean of the target variable $\yt[n]$ as a first-order approximation of the parameters centred around their previous estimate.
    Let $\sigma(x) = (1 + \exp(-x))^{-1}$ be the sigmoid function and $\sigma'(x) = \sigma(x) (1 - \sigma(x))$ its derivative.
    Then, an approximate solution to \eqref{eq:subspace-last-rvga}
    is given by
    \begin{align}
        \mlast_n &= \mlast_{n-1} + \covlast_{n-1}\,h(\mhidden_{n-1}; \xt[n])  \Big(\yt[n] - \sigma(\mlast_{n-1}^\intercal\, h(\mhidden_{n-1}; \xt[n]))\Big), \label{eq:part-mlast-update} \\
        \covlast_{n}^{-1} &= \covlast_{n-1}^{-1} + \sigma'\big(\mlast_{n-1}^\intercal\, h(\mhidden_{n-1}; \xt[n])\big) h(\mhidden_{n-1}; \xt[n])^\intercal\, h(\mhidden_{n-1}; \xt[n]), \label{eq:part-covlast-update}\\
        \mhidden_t &= \mhidden_{n-1} + \covhidden_{n-1}\nabla_{\phiddensub}h(\mhidden_{n-1}; \xt[n])  \Big(\yt[n] - \sigma(\mlast_{n-1}^\intercal\, h(\mhidden_{n-1}; \xt[n]))\Big), \label{eq:part-mhidden-update}\\
        \covhidden_{n}^{-1} &= \covhidden_{n-1}^{-1} + \sigma'\big(\mlast_{n-1}^\intercal\,
        h(\mhidden_{n-1}; \xt[n])\big) \nabla_{\phiddensub} h(\mhidden_{n-1}; \xt[n]) \nabla_{\phiddensub}h(\mhidden_{n-1}; \xt[n])^\intercal \,, \label{eq:part-covhidden-update}
    \end{align}
    where $\mhidden_{n}$, $\covhidden_{n}$ are the estimated mean and covariance of the projected-hidden-layer parameters at step $n$, and
    $\mlast_n$, $\covlast_n$ are the estimated mean and covariance matrix of the last-layer parameters.
\end{theorem}

As a corollary, a variant of {\ourmodel} can be derived
to model any other member of the exponential family by replacing
the mean and covariance of the target distribution of choice.
\noindent See \ref{app: ourmodel derivations} for a proof of Theorem \ref{theorem:subspace-last-rvga}.

\section{Asynchronous learning and decision making}
\label{sec:deployment-methods}
Next, we discuss how we deploy and evaluate the performance of the online {\ourmodel} methodology with asynchronous data.
In classical filtering problems, 
as soon as new information arrives at, say, time $t_i$, the parameters $\bm\theta_{t_i}$ of the model are updated. Next, when a new trade arrives at time $t_{i+1}$,
one uses the parameters $\bm\theta_{t_i}$ to estimate if a trade will be toxic.
In our setting, however, an update at time $t_{i+1}$ with $\bm\theta_{t_i}$ is only possible if
$t_{i+1}$ is greater than the time of last trade $t_i$ plus the toxicity horizon $\valuetoxichorizon > 0$, i.e.,
$t_{i + 1} > t_i + \valuetoxichorizon$. Otherwise, we use $\bm\theta_{t_j}$ to predict the probability of a toxic trade, with $j = \argmax_k t_{i + 1} > t_k + \valuetoxichorizon$.
Figure \ref{fig:diag-async-train-predict} illustrates this procedure.

\vspace{0.5cm}

\begin{figure}[H]
    \centering
    \begin{tikzpicture}
\draw[ultra thick, ->] (0, 0) -- (12cm, 0);

\foreach[count=\i] \x in {0, 1, 2.75, 4, 4.25, 4.5, 7.5}
\draw (\x cm, 3pt) -- (\x cm, -3 pt) node[below=3pt]{\scriptsize $t_\i$};

\foreach[count=\i] \x in {0, 1, 2.75, 4, 4.25, 4.5, 7.5}
\draw[*-, cyan] (\x, 0.5 + 0.2 * \i) -- (\x + 2.5, 0.5 + 0.2 * \i);

\foreach[count=\i] \x in {0, 1, 2.75, 4}
\draw[purple, densely dotted] (\x + 2.5, 2) -- (\x + 2.5, -1) node[below] {\scriptsize $\,\bm\theta_\i\,$};
\draw[purple, densely dotted] (0, -1) node[below] {\scriptsize $\bm\theta_0$};

\foreach[count=\i] \x in {4.25}
\draw[purple, densely dotted] (\x + 2.5, 2) -- (\x + 2.5, -1.25) node[below] {\scriptsize $\,\bm\theta_5\,$};

\foreach[count=\i] \x in {4.5}
\draw[purple, densely dotted] (\x + 2.5, 2) -- (\x + 2.5, -1) node[below] {\scriptsize $\,\bm\theta_6\,$};

\foreach[count=\i] \x in {7.5}
\draw[purple, densely dotted] (\x + 2.5, 2) -- (\x + 2.5, -1) node[below] {\scriptsize $\,\bm\theta_7\,$};
\end{tikzpicture}
    \vspace{-1cm}
    \caption{
    Asynchronous predict-update steps: trades arrive at irregular times $\{t_i\}_i$.
    An update to the model is only possible if $t_{i+1} > t_i + \valuetoxichorizon$, when we know
    whether the trade was toxic or benign. In this example, the model parameters $\bm\theta_0$
    are known at time $t_1$, when a new trade arrives. When a second trade arrives, at time $t_2$, we
    do not know whether the previous trade was toxic or benign at $t_1$, so we use the model weights $\bm\theta_0$
    to make a prediction. The next trade arrives at time $t_3 > t_1 + \valuetoxichorizon$, so we use $\bm\theta_1$ to make a prediction.
    Finally, multiple trades arrive consecutively at times $t_4$, $t_5$, and $t_6$, in which case we use
    $\bm\theta_2$. The last trade to arrive at time $t_7$ uses $\bm\theta_6$. In this example
    $\bm\theta_3$, $\bm\theta_4$, and $\bm\theta_5$ were never used to make a prediction because 
    {
    trades did not arrive during the period $[t_3 + \valuetoxichorizon, t_6 + \valuetoxichorizon)$.
    }
    }
    \label{fig:diag-async-train-predict}
\end{figure}

We employ the asynchronous online updating for \ourmodel\  and MLE and select hyperparameters over the warmup stage. 
Figure \ref{fig:pulse-description-ts} shows how {\ourmodel} updates model parameters based on:
current model parameters $\bm \theta$, features $\bf x$, and outcome $y$.
Here, $\bm \theta =  (\phidden, \plast)$ are the model parameters one uses to produce $p(y=1 \,\vert\, {\bf x}, {\bm\theta})$, with which we compute the prediction $\hat y$.
We employ the predictions $\hat{y}$ and the outcomes $y$ to compute the accuracy defined above.

\begin{figure}[H]
\centering
\scalebox{0.8}{\begin{tikzpicture}[scale=1]
    \tikzstyle{short_sym} = [ shorten >=20, shorten <=20 ]
    
    \draw[ultra thick, ->] (-3, -1.2) -- (4, -1.2);
    
    \draw
    (-3, -1.2cm - 5pt) --
    (-3, -1.2cm + 5pt)
        node[below, yshift=-10pt] {$s$};
    
    \draw
    (0, -1.2cm - 5pt) --
    (0, -1.2cm + 5pt)
    node[below, yshift=-10pt] {$s + \valuetoxichorizon$};

    \draw
    (3, -1.2cm - 5pt) --
    (3, -1.2cm + 5pt)
        node[below, yshift=-10pt] {$t$};

    \draw[thick, short_sym, ->]
    (-3, 2.5) node[below] {${\bf x}_s$}
    to[out=230, in=90]
    (-3, -1);
    
    \draw (-3, 1) node[] {$\boldsymbol{\theta}_s$};
    
    \draw[thick, short_sym, ->]
    (-3, 1)
    --
    (-3, -1) node[above] {$\hat{y}_s$};

    \draw[thick, short_sym, ->]
    (-3, -1)
    --
    (0, 1);
    
    \draw (0, 1) node {$\boldsymbol{\theta}_{s + \valuetoxichorizon}$};
    
    \draw[thick, short_sym, ->]
    (0, -1) node[above] {$y_s$}
    --
    (0, 1);
    
    \draw[thick, short_sym, ->]
    (-3, 1)
    --
    (0, 1);
    
    \draw[thick, short_sym, ->]
    (3, 2.5) node[below] {${\bf x}_t$}
    to[out=230, in=90]
    (3, -1);
    
    \draw (3, 1) node {$\boldsymbol{\theta}_t$};
    
    \draw[thick, short_sym, ->]
    (3, 1)
    --
    (3, -1) node[above] {$\hat{y}_t$};
    
    \draw[thick, short_sym, ->, dashed]
    (0, 1)
    --
    (3, 1);    
\end{tikzpicture}}
\caption{
    {\ourmodel} update procedure. For simplicity, we take $s+\valuetoxichorizon < t$. 
}
\label{fig:pulse-description-ts}
\end{figure}

\subsection{Model for decision making}\label{sec: model for decision making}
Here, we devise brokerage strategies that employ predictions of toxic flow.  
We introduce a  one-shot optimisation problem for the broker's strategy to internalise-externalise trades.

For method $\text{M} \in \{\text{\ourmodel}, \text{ LogR}, \text{ RF}, \text{ MLE}\}$, let
$p^{+,\text{M}}\in (0,1)$, denote the probability that a buy order will be toxic
and let $p^{-,\text{M}}$ denote the probability that a sell order will be toxic.
Note that $p^{+,\text{M}} + p^{-,\text{M}}$ does not necessarily add to $1$.
Let $\spread/2 >0$ denote the half bid-ask spread and let $\eta>0$ denote the shock to the midprice $S$ if the trade is toxic;
here we assume that $\eta\in(\spread, \infty)$.
The broker controls $\delta^{\pm}\in\{0,1\}$. When $\delta^{\pm} = 0$  the broker externalises the trade
and when $\delta^{\pm} = 1$ the broker internalises the trade.
The inventory of the broker is $Q\in \mathbb{R}$;
when $Q > 0$ the broker is long and
when $Q < 0$ the broker is short. Assume all trades are for one unit of the asset.
Then, the broker solves
\begin{align}\label{eq:trading-problem}
   \delta^{\pm *} = \argmax_{\delta^{\pm}\{0,1\}}
   \mathbb{E}\left[
   \underbrace{
   \overbrace{\pm\delta^{\pm}\,(S \pm \spread/2)}^\text{cash flow} + 
   \overbrace{(S \pm \eta\,Z)\left(Q \mp \delta^{\pm}\right)}^\text{inventory valuation}
   }_\text{mark-to-market}
   \underbrace{ -
   \phi\,\left(Q\mp \delta^\pm\right)^2 }_{\text{inventory penalty}}
   \right]\,,
\end{align}
where $Z$ is a Bernoulli random variable with parameter $p^{\pm,\text{M}}$, and $\phi\geq 0$ is an inventory penalty parameter.
Intuitively, the broker optimises the expected mark-to-market value of her portfolio after internalising the trade adjusted by a quadratic penalty on inventory.
The solutions to \eqref{eq:trading-problem} are
\begin{equation}\label{eq: optimal strategy delta^pm}
    \delta^{\pm *} = \mathds{1}\left(\frac{\spread/2}{\eta} - \frac{\phi}{\eta} \pm \frac{2\,\phi}{\eta}\,Q  > p^{\pm,\text{M}}\right) = \mathds{1}\left( \cutoff \pm \Phi\,Q > p^{\pm, \text{M}}\right)\,,
\end{equation}
where $\cutoff :=  {\spread}/{2\,\eta} - \phi/{\eta}$ and $\Phi:= 2\,\phi/\eta$.
We call $\Phi$ the inventory aversion parameter and we call $\cutoff $ the {\cutoffprobability}. The strategy internalises trades when the prediction $p^{\pm, \text{M}}$ is lower than the {\cutoffprobability} $\cutoff$ adjusted by the inventory of the broker $Q$ and the inventory  aversion parameter $\Phi$.
When either
\begin{equation*}
    \cutoff + \Phi\,Q  = p^{+,\text{M}}\qquad \text{ or }\qquad \cutoff - \Phi\,Q  = p^{-,\text{M}}\,,
\end{equation*}
the broker is indifferent between internalising or externalising the trade in the market; this happens with probability zero.

{
The model variables $S$, $\eta$, and $\phi$ in \eqref{eq:trading-problem} are not calibrated in our experiments. 
Their role here is to motivate the broker’s decision rule \eqref{eq: optimal strategy delta^pm} in terms of the cutoff probability $\cutoff$ and the inventory aversion parameter $\Phi$. 
}

Next, we study the case $\Phi =0$ in more detail. In \ref{app: amb aversion} we explore the case $\Phi >0$.

\subsection{Internalise-externalise strategy}
Motivated by the mathematical framework in Subsection \ref{sec: model for decision making}, below we introduce a family of predictions of toxicity based on the {\cutoffprobability} $\cutoff\in [0,1]$. In what follows, we ignore the permanent price impact of externalising trades, as this  would require the formulation of a stochastic control problem.\footnote{ 
For a series of recent related papers on broker-client stochastic control problems see \cite{cartea2022brokers,bergault2025mean,cartea2024nash,donnelly2025liquidity,wu2024broker}. In these papers, the externalisation activity of the broker has a permanent price impact in prices.  There, trades are toxic because informed traders observe the stochastic drift of the asset price and exploit that information optimally. The problem then is to characterise the equilibrium between the broker and the informed traders (the uninformed trader is usually assumed to be non-strategic). 
One can also allow the broker to have flexibility in the skewing mechanism,
as opposed to using bid and ask quotes that are aligned to those in the lit market,
see e.g., \cite{barzykin2023algorithmic,barzykin2022market}. Furthermore, the broker can employ filtering techniques to identify toxic trades. For example, if skewed quotes which are ``expensive'' are hid by a client,
this may convey information about the toxicity of the trade;
see e.g., \cite{aqsha2024strategic} where the authors filter the trading speed of informed traders in search of the signal. } 
Below, when deploying our strategies, the historical data do not change.

\begin{definition}[$\cutoff$-predicted toxic trade]
Let $\cutoff \in [0,1]$ and $p(y = 1\,\vert\, \xt[t_n], {\bm\theta})$ be the output of a classifier.
A trade is predicted to be toxic with {\cutoffprobability} $\cutoff$ if
\begin{equation}
    p(y=1 \,\vert\, \xt[t_n], {\bm\theta}) > \cutoff\,.
\end{equation}
We store the decision of a toxic trade in the variable
\begin{equation}
    \hat y_{t_n}^\cutoff = \mathds{1}(p(y=1 \,\vert\, \xt[t_n], {\bm\theta}) > \cutoff)\,.
\end{equation}

\end{definition}

We are interested in the predictive performance of the models as we vary the value of $\cutoff$.
To this end, let $y_{t_n} \in \{0, 1\}$ denote if a trade executed at time $t_n$ was toxic at time $t_n+\valuetoxichorizon$
($y_{t_n}=1$ if toxic and $y_{t_n}=0$ otherwise).
We employ the true positive rate and the false positive rate, which we define below.

\begin{definition}\label{def:tpr}
    The true positive rate (TPR) of a sequence of trades $\{y_{t_n}\}_{n=1}^N$ with predictions $\{\hat y_{t_n}\}_{n=1}^N$
 at a {\cutoffprobability} $\cutoff$ is
    \begin{equation}
    \begin{aligned}
        \text{TPR}_\cutoff
        &= \frac
            {\sum_{n=1}^N\mathds{1}(y_{t_n} = \hat{y}_{t_n}^\cutoff) \cdot \mathds{1}(y_{t_n} = 1)}
            {\sum_{n=1}^N\mathds{1}(y_{t_n} = 1)}\,.
    \end{aligned}
    \end{equation}
\end{definition}

\begin{definition}\label{def:fpr}
    The false positive rate (FPR) of a sequence of trades $\{y_{t_n}\}_{n=1}^N$ with predictions $\{\hat y_{t_n}\}_{n=1}^N$
 at a {\cutoffprobability} $\cutoff$ is
    \begin{equation}
    \begin{aligned}
        \text{FPR}_\cutoff
        &= \frac
            {\sum_{n=1}^N\mathds{1}(y_{t_n} \neq \hat{y}_{t_n}^\cutoff) \cdot \mathds{1}(y_{t_n} = 0)}
            {\sum_{n=1}^N\mathds{1}(y_{t_n} = 0)}\,.
    \end{aligned}
    \end{equation}
\end{definition}

Each choice of $\cutoff$ induces a pair of values $(\text{FPR}_\cutoff, \text{TPR}_\cutoff)$. The graph of $\cutoff\to(\text{FPR}_\cutoff, \text{TPR}_\cutoff)$ is known as
the Receiver Operating Characteristic (ROC). Figure \ref{fig:roc-30s} shows the daily
ROC of the models in the deploy stage with {\toxichorizon} of 30s.

\begin{figure}[H]
    \centering
    \includegraphics[width=0.6\linewidth]{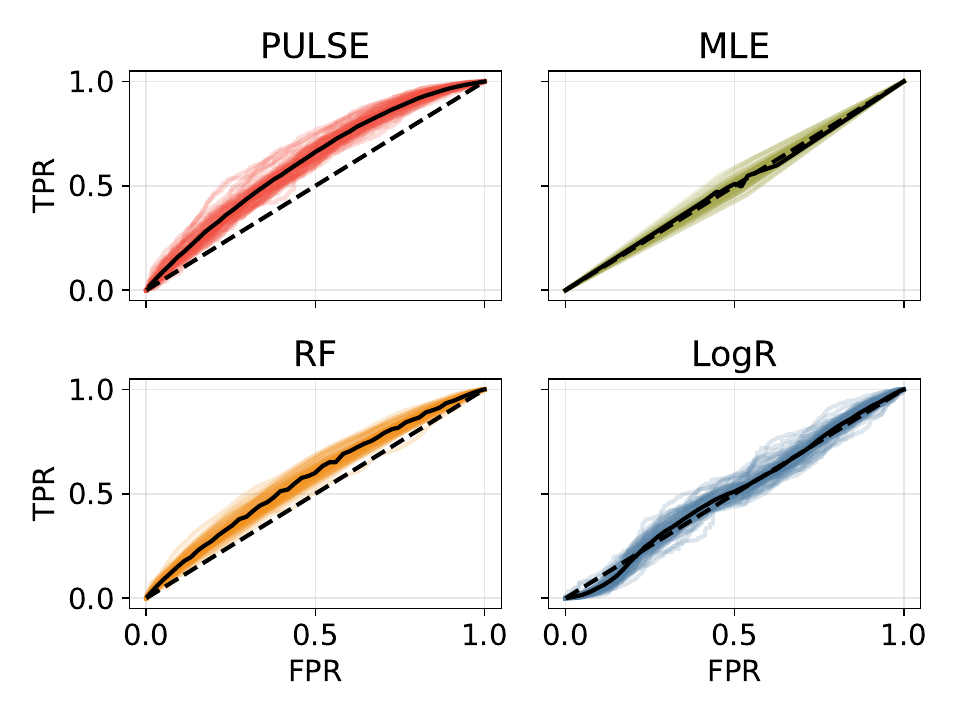}
    \caption{
    Daily ROC curves with {\toxichorizon} of 30s. We plot the daily ROC curve for each model
    at the end of each trading day. Each coloured line represents the ROC
    curve for a trading day. The solid black line is the average of the daily
    ROC curves. Finally, the black dashed line represents the ROC curve for
    a random classifier.
    }
    \label{fig:roc-30s}
\end{figure}

The area under an ROC curve, called AUC, is used in the machine learning literature to compare
classifiers; see, e.g., \cite{fawcett-roc}.
Intuitively, the AUC is a measure to quantify a classifier's ability to distinguish
between toxic and benign trades.

\section{Experiments}
\label{section:traders-real}
We employ the methodology developed in the previous section with the following configuration.
The NNet for {\ourmodel} is an MLP with three hidden layers, 100 units in each layer, and ReLU activation function. The
number of epochs $E$ is 850, we 
skip the first $50$ iterations of the optimisation procedure, the
subspace dimension is $d=20$, the
learning rate is $\alpha = 10^{-7}$, and we
store gradients every $k=4$ steps.\footnote{Our results are robust to higher subspace dimensions. In particular, we find that using a dimension up to $d=1000$ does not yield a statistically significant outperformance over the baseline with $d=20$.}
With this configuration, we estimate the  $38,700$ units of the MLP  during the warmup stage. For the deployment stage,
{\ourmodel} updates $120$ degrees of freedom;
this accounts for less than half a percent of all parameters updated during the warmup stage.
From a practical perspective,
the memory cost of a single step of the algorithm (during the deployment stage)  is $O((L + D)^2)$.
In this paper, an update requires less than 1mb of memory because each unit is a 32bit float.
Conversely, if we did not employ the subspace approach
a single step would require 190gb of memory, making it infeasible to run on typical GPU devices.

We divide the dataset into a warmup dataset ${\cal D}_\text{warmup}$ and a deploy dataset ${\cal D}_\text{deploy}$. 
Here, ${\cal D}_\text{warmup}$ is from 28 June 2022 to 29 July 2022,
and ${\cal D}_\text{deploy}$ is from
1 August 2022 to 21 October 2022.

\subsection{Benchmark models}
We compare the performance of four methods: {\ourmodel}, logistic regression (LogR), random forests (RF),
and a recursively updated maximum-likelihood estimator of a Bernoulli-distributed random variable (MLE).
The MLE benchmark reflects a common industry practice in which toxicity is treated as a client-level attribute rather than a trade-specific one.
We include LogR as a widely used linear baseline that helps quantify the added value of employing non-linear models such as neural networks.
Finally, we include RF, a strong tree-based nonparametric method that is competitive on tabular data \citep[e.g.,][]{mcelfresh2023neural}.

With logistic regression, the probability that a trade is toxic is
\begin{equation}\label{eq:logistic-regression}
    p(y \,\vert\, \plast_{0};\, \xt)
    = \text{Bern}\Big(y \,\vert\, \sigma(\plast_0^\intercal\, \xt)\Big),\qquad y\in\{0,1\}\,,
\end{equation}
where $\plast_0$ is estimated maximising the log-likelihood with L-BFGS; see \cite{liu89-lbfgs}.

Next, RF is a bootstrap-aggregated collection of de-correlated trees.
To predict if a trade is toxic one uses the average over the individual trees in the ensemble,
see Section 15.1 of \cite{hastie01-esl}.

Further, we model the unconditional probability  of a toxic trade as a Bernoulli-distributed random variable
with mass function
\begin{equation}\label{eq:bernoulli}
    p(y \,\vert\, \pi)
    = \text{Bern}\Big(y \,\vert\, \pi\Big),\qquad y\in\{0,1\}\,.
\end{equation}
The maximum-likelihood estimator of the parameter $\pi$, given a collection $\{z_1, \ldots, z_N\}$
of Bernoulli-distributed samples, where $z_n \in \{0, 1\}$, is given by
\begin{equation*}
    \pi_\text{MLE} = \frac{1}{N}\,\sum_{n=1}^N \mathbbm{1}(z_n = 1)\,;
\end{equation*}
here, $\mathds{1}(\cdot)$ is the indicator function and
we refer to this estimator as the MLE method.
This quantity is updated after each new observation $\yt$.

The decision rule in \eqref{eq: optimal strategy delta^pm} is directional and therefore induces two conditional problems:
one for buy trades and one for sell trades.
We consequently learn side-specific feature transformations and parameters. There are three implementations:
\begin{align}
& {\bf x}_\text{bid} \to {\cal M}_{\text{bid}} \to y,\;\; {\bf x}_\text{ask} \to {\cal M}_{\text{ask}} \to y,
\label{eq:choice-model}\\
& [{\bf x}_\text{bid},\, {\bf x}_\text{ask}] \to {\cal M} \to [y_{\text{bid}},y_{\text{ask}}],\\
& [{\bf x},\text{bid/ask}] \to {\cal M} \to y,
\end{align}
where ${\cal M}$ is for model.
We adopt \eqref{eq:choice-model} because it  enables a direct comparison with the baselines (LogR, RF, MLE).

In \ref{app-section:value-of-single-model}  we study the performance when we train one model per client
and use the client's unique features.

\textbf{Online evaluation.}
We evaluate MLE and {\ourmodel}  online because both admit single-pass updates compatible with asynchronous labels.
Our pipeline also supports online logistic regression;
however, in our data, LogR underperforms even with batch retraining, yielding score distributions concentrated near the base rate,
so online updates would not improve its behaviour.
Fully online RF would require incremental trees or replay buffers which require more compute power takes longer to run. 
For robustness, in \ref{sec:weekly-retrains} we re-train LogR and RF weekly on accrued data.

\subsection{Model comparison}
In this section, we analyse the AUC of the methods for various \toxichorizon{s}.
For each method, we compute the AUC for the sequence of trades of each day
and show a density plot of such values in Figure \ref{fig:auc-bday-btoxic-horizon};
recall that the deploy window is between 1 August 2022 and 21 October 2022.

\begin{figure}[H]
    \centering
    \includegraphics[width=0.45\linewidth]{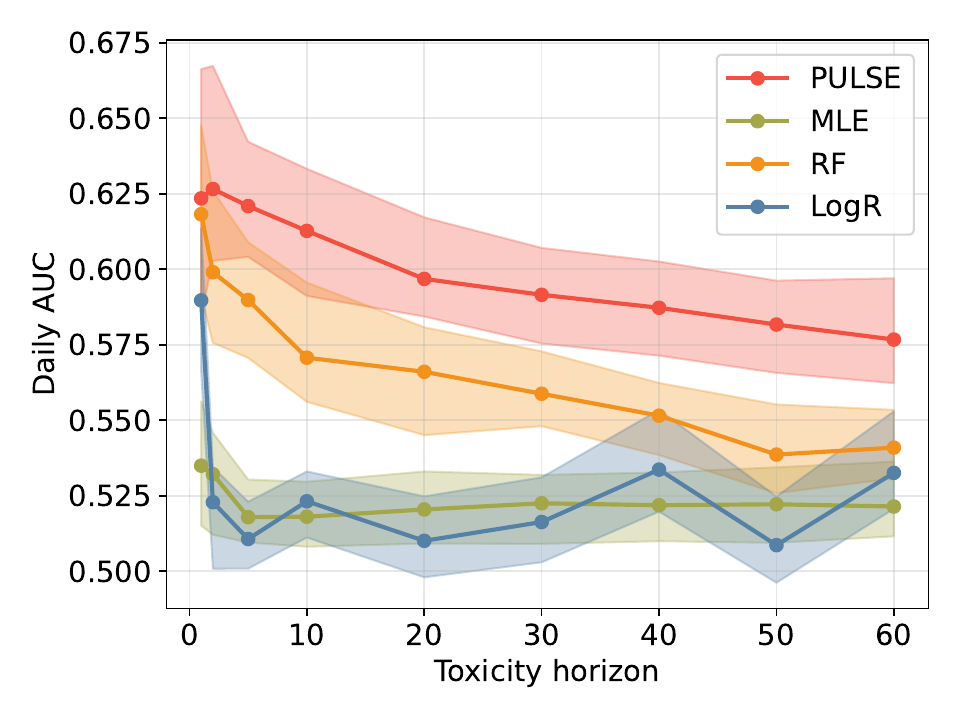}
    \includegraphics[width=0.45\linewidth]{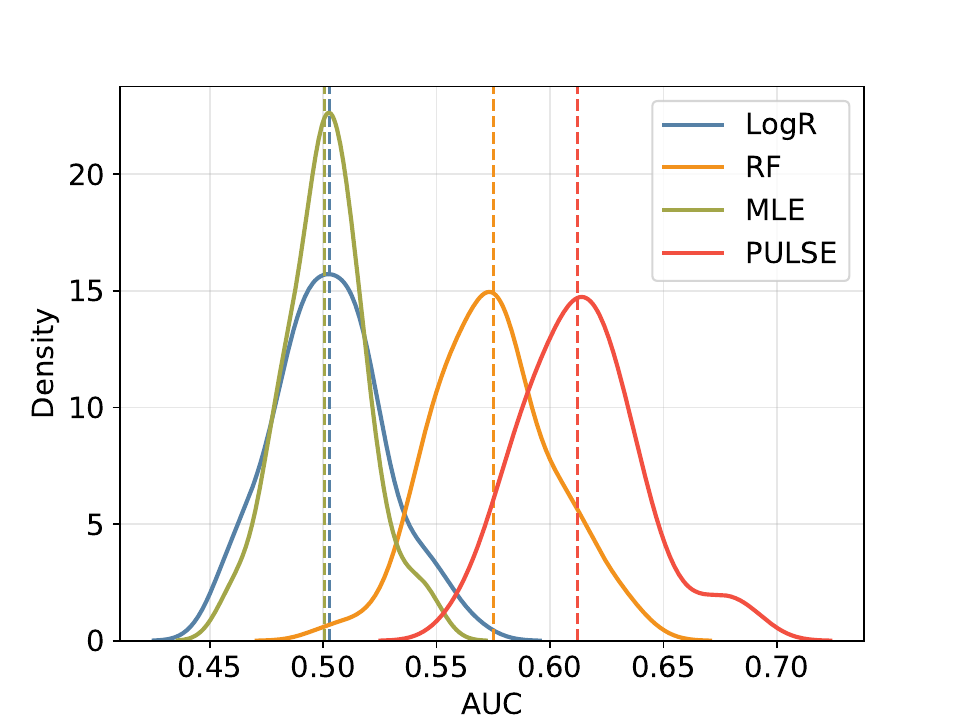}
    \caption{
    Left panel shows the median (line) and interquartile range (shaded) of daily AUC by toxicity horizon and model.
    Right panel shows  the daily AUC of all models for a {\toxichorizon} of 30s. The dashed vertical lines correspond to the mean daily value.
    EUR/USD currency pair over the period 1 August 2022 to 21 October 2022.
    }
    \label{fig:auc-bday-btoxic-horizon}
\end{figure}

\ourmodel\ has the highest average AUC among the four methods
for all \toxichorizon{s}. Furthermore, as the toxic
horizon increases, the AUC of {\ourmodel} and RF decrease but the outperformance of {\ourmodel} over RF 
increases. MLE and LogR have the poorest performance.
We observe that the AUC for PULSE and RF decreases as the toxicity horizon increases, 
reflecting  that the prediction problem is more difficult over longer horizons. 
Empirically, the variance of profitability increases with time (see Figure \ref{fig:sharpness-profile}), 
which reduces the predictive power of these methods. 
For MLE, the AUC remains close to 0.5 across horizons, 
highlighting that it acts as a random classifier.
On the other hand, LogR  shows high variability across horizons, 
mainly because its linear structure is not well-suited to capture the dynamics of the data, 
leading to poor calibration and unstable performance.

Figure \ref{fig:auc-hist-models} shows the five-day exponentially-weighted average
of the AUC for each day. RF, PULSE, and MLE attain their maximum values
at the beginning of the deploy period and then decay over time, while
{\ourmodel} maintains a steady performance
 because {\ourmodel} updates its parameters with each new
observation.
The performance of MLE and LogR is similar.
In our data, LogR does not find a useful linear boundary.
Consequently, LogR's predicted probabilities cluster near the unconditional
rate and are highly correlated with MLE’s, yielding little time variation.

\begin{figure}[H]
    \centering
    \includegraphics[width=0.6\linewidth]{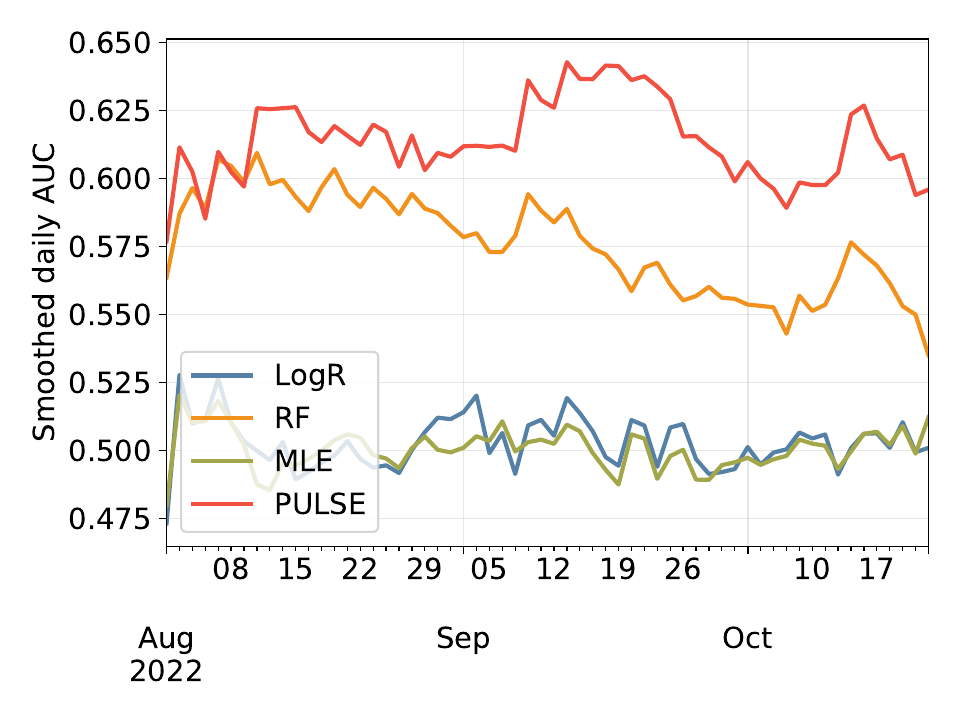}
    \caption{
    Five-day exponentially-weighted moving average of AUC over time. The {\toxichorizon} is 30 seconds. 
    }
    \label{fig:auc-hist-models}
\end{figure}
Figure \ref{fig:pulse-auc-hist-models} shows a five-day exponentially-weighted moving average (with decay $1/3$) 
of the AUC for the various {\toxichorizon}s across time.
\begin{figure}[H]
    \centering
    \includegraphics[width=0.6\linewidth]{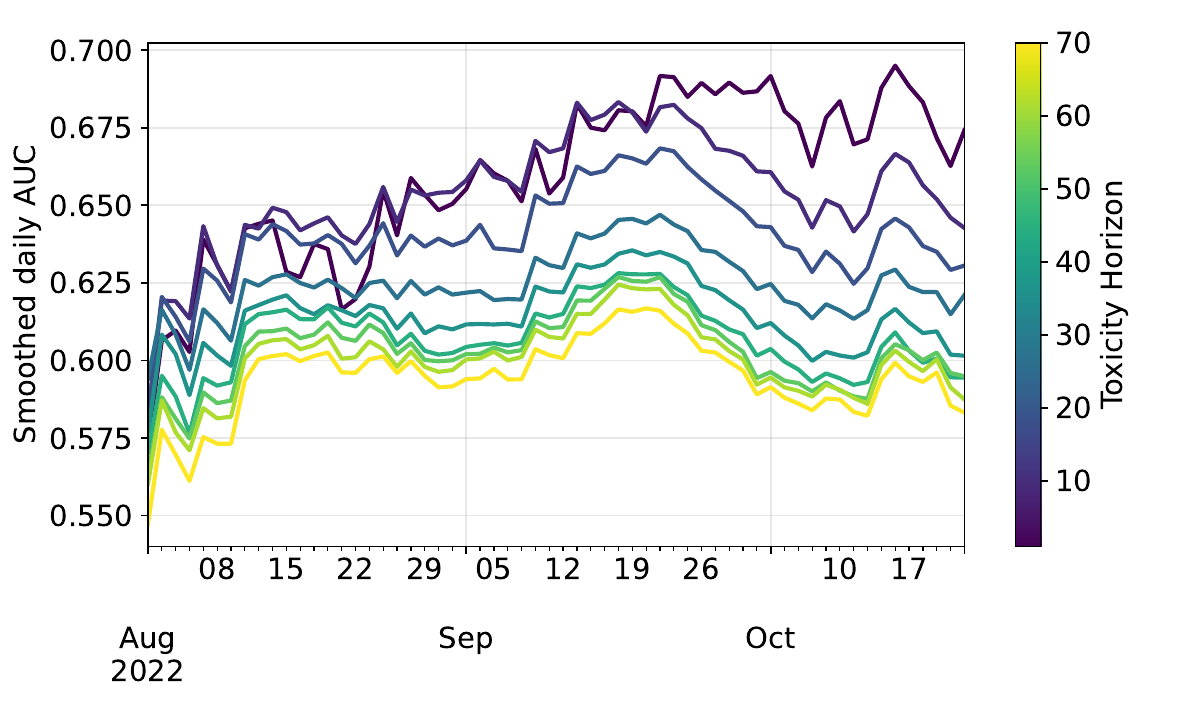}
    \caption{
    Five-day exponentially-weighted moving average of AUC over time for {\ourmodel} across {\toxichorizon}s.
    }
    \label{fig:pulse-auc-hist-models}
\end{figure}
We observe that as the {\toxichorizon} increases, the maximum AUC decreases almost uniformly.

\subsection{Trade Prediction Effectiveness and Missed Opportunities}
\label{sec:trade-prediction}
We employ data for the EUR/USD currency pair over the period 1 August 2022 to 21 October 2022 to test the internalisation-externalisation strategy.
As above,
we ignore inventory aversion (i.e., $\Phi=0$).

In what follows, we define the avoided profit of an externalised trade to be the PnL of that trade had it been internalised with unwinding at the end of the {\toxichorizon}.
Figure \ref{fig:pnl-cross-avoided-loss} reports the PnL ($y$-axis) of the internalised trades and the avoided profit ($x$-axis) of the externalised trades when the broker uses the internalisation-externalisation strategy  
\eqref{eq: optimal strategy delta^pm}  and for {\ourmodel}, MLE, LogR, and RF.
The points shown for each method and {\toxichorizon} are the ones that maximised PnL over all possible cutoff probabilities $\cutoff\in \{0.05,0.15,0.25, \ldots, 0.95\}$.

The broker starts with zero inventory at the beginning of 1 August 2022 and she crosses the spread to unwind the internalised trades at the end of the {\toxichorizon}.
We keep track of the PnL they would have obtained over the toxicity horizon.\footnote{Studying the variance of these results is more challenging and left for future research.}
The inventory is in euros (\euro) and the PnL is in dollars (\$).
Each trade is for the median quantity which is \euro 2,000 in our dataset.
When a trade is toxic, the median loss to unwind the position is  $\$ 7\times 10^{-5}$ per  euro traded and when the trade is not toxic, the median profit is  $\$ 8\times 10^{-5}$ per euro traded.

\begin{figure}[H]
    \centering
    \includegraphics[width=0.6\linewidth]{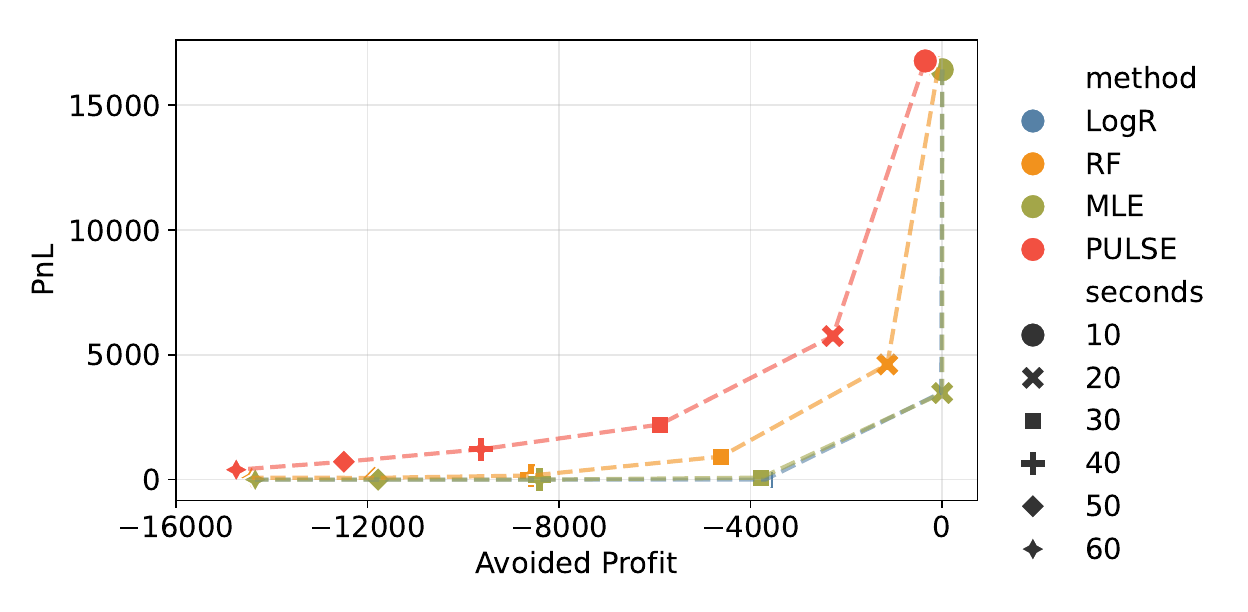}
    \caption{
        PnL and avoided profit for various {\toxichorizon}s  and for {\ourmodel}, MLE, LogR, and RF. Each point shows the highest possible PnL for a given method and {\toxichorizon} where the maximum is taken across cutoff probability where $\cutoff\in \{0.05,0.15,0.25, \ldots, 0.95\}$.
    }
    \label{fig:pnl-cross-avoided-loss}
\end{figure}

For each {\toxichorizon}, the internalisation-externalisation strategy informed by \ourmodel attains the highest PnL and the lowest
avoided profits,\footnote{{A negative avoided profit is an incurred  loss had the trade been internalised and unwound at the toxicity horizon.}}
see red dots joined by the dash line.
{These results show the added economic value that one obtains when informing trading strategy \eqref{eq: optimal strategy delta^pm} with the predictions made by {\ourmodel}. Indeed, we show that the higher quality of predictions obtained by {\ourmodel} produce higher PnL and lower avoided profits. }

Next,
Figure \ref{fig:predicted-probas} shows a histogram of $p^{\pm,\,\text{M}}$ for $\text{M}\in\{\text{{\ourmodel}}, \text{MLE}, \text{LogR}, \text{RF}\}$ and {\toxichorizon} of 10s
and
Figure \ref{fig:pct-int-vol} shows how the percentage of internalised volume depends on the  {\cutoffprobability}.
\begin{figure}[H]
    \centering
    \includegraphics[width=0.8\linewidth]{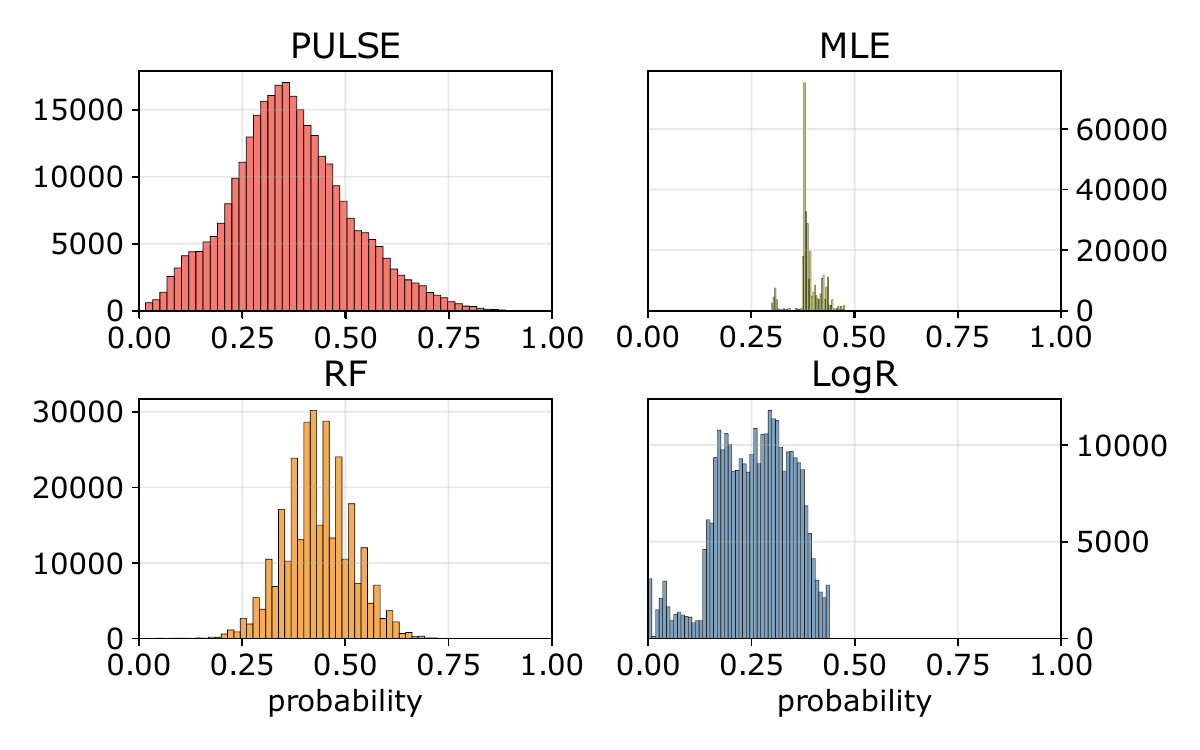}
    \caption{
    Histograms of predicted toxicity probabilities by model (pooled across clients). 
    Toxicity horizon is ten seconds.
    EUR/USD currency pair over the period 1 August 2022 to 21 October 2022.
    }
    \label{fig:predicted-probas}
\end{figure}

{
MLE produces an almost constant score (around the unconditional toxicity rate).
For LogR, the concentration of probabilities below $0.5$
is due to the nonlinear structure of the data.
Consequently, when the cutoff crosses these concentrated values, the internalised volume changes abruptly.
In contrast, {\ourmodel} produces a broader,
more dispersed (higher-entropy) distribution of probabilities, enabling smoother and more informative policy adjustments.
}

\begin{figure}[H]
    \centering
    \includegraphics[width=0.5\linewidth]{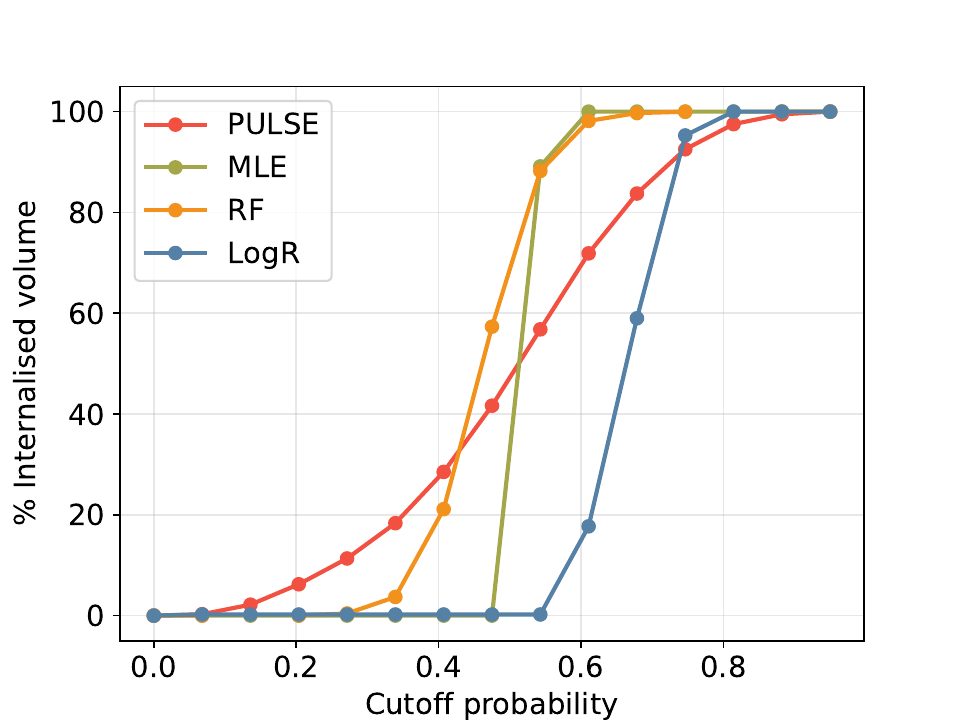}
    \caption{
        Percentage of internalised volume as a function of the  {\cutoffprobability} $\cutoff$
        for {\toxichorizon} of 20s.
        EUR/USD currency pair over the period 1 August 2022 to 21 October 2022.
    }
    \label{fig:pct-int-vol}
\end{figure}

{
Figure~\ref{fig:pct-int-vol} plots the percentage of internalised volume as a function of the cutoff $\cutoff$.
Step-like moves occur when $\cutoff$ crosses regions where scores are concentrated.
For MLE, scores are almost constant at the unconditional toxicity rate (as seen in the previous histogram), thus, the percentage of internalised volume jumps when $\cutoff$ crosses that level.
For LogR, scores cluster near the base rate, producing a sharp transition similar to that of MLE.
RF shows smaller but noticeable jumps, reflecting its narrower score dispersion.
In contrast, {\ourmodel} produces a more dispersed score distribution and therefore a smoother curve. See \ref{sec: precision recall} for an analysis on precision and recall.
}

\section{Conclusions}
\label{sec:conclusions}
We employed  machine learning and statistical methods to detect toxic flow.
We also developed a novel method for online training of neural networks, which we call {\ourmodel}.
We use {\ourmodel} to estimate the parameters of a neural network (sequentially)
that computes the probability that an incoming trade will be toxic. 
The out-of-sample performance of the multilayered-perceptron (MLP) trained with {\ourmodel} is high, stable, and outperforms
the other methods we considered.

We proposed a broker's strategy that uses these predictions to decide which trades are internalised and which are externalised by the broker.
The mean PnL of the internalise-externalise strategy we obtain
when training the MLP with {\ourmodel}
is the highest (when compared with the benchmarks)
and it is robust to model parameter choices.
Future research will consider a hierarchical version of the problem, where there is a structure for toxicity common to
all traders. 
The methodology proposed by {\ourmodel}
can also be used in other areas of finance
where sequential updates are desirable,
such as in the prediction of fill-rate probabilities, and in multi-armed bandit problems for trading; see, e.g., \cite{arroyo2023deep} and \cite{cartea2023bandits}.

\section*{Comments}
\noindent For the purpose of open access, the authors have applied a CC BY public copyright licence to any author accepted manuscript version arising from this submission.

\section*{Funding}
\noindent  No funding was received.

\section*{Disclosure of interest}
\noindent  There are no interests to declare.

\bibliographystyle{apalike}
\bibliography{References}

@article{west2022best,
  title={Best practice in statistics: The use of log transformation},
  author={West, Robert M},
  journal={Annals of clinical biochemistry},
  volume={59},
  number={3},
  pages={162--165},
  year={2022},
  publisher={SAGE Publications Sage UK: London, England}
}

@article{mcelfresh2023neural,
  title={When do neural nets outperform boosted trees on tabular data?},
  author={McElfresh, Duncan and Khandagale, Sujay and Valverde, Jonathan and Prasad C, Vishak and Ramakrishnan, Ganesh and Goldblum, Micah and White, Colin},
  journal={Advances in Neural Information Processing Systems},
  volume={36},
  pages={76336--76369},
  year={2023}
}

@article{barzykin2022market,
  title={Market-making by a foreign exchange dealer},
  author={Barzykin, Alexander and Bergault, Philippe and Gu{\'e}ant, Olivier},
  journal={Risk (Cutting Edge)},
  year={2022}
}

@article{bergault2025mean,
  title={A mean field game between informed traders and a broker},
  author={Bergault, Philippe and S{\'a}nchez-Betancourt, Leandro},
  journal={SIAM Journal on Financial Mathematics},
  volume={16},
  number={2},
  pages={358--388},
  year={2025},
  publisher={SIAM}
}

@article{barzykin2023algorithmic,
  title={Algorithmic market making in dealer markets with hedging and market impact},
  author={Barzykin, Alexander and Bergault, Philippe and Gu{\'e}ant, Olivier},
  journal={Mathematical Finance},
  volume={33},
  number={1},
  pages={41--79},
  year={2023},
  publisher={Wiley Online Library}
}

@article{aqsha2024strategic,
  title={Strategic learning and trading in broker-mediated markets},
  author={Aqsha, Alif and Drissi, Fay{\c{c}}al and S{\'a}nchez-Betancourt, Leandro},
  journal={arXiv preprint arXiv:2412.20847},
  year={2024}
}

@article{wu2024broker,
  title={Broker-Trader Partial Information Nash-Equilibria},
  author={Wu, Xuchen and Jaimungal, Sebastian},
  journal={arXiv preprint arXiv:2412.17712},
  year={2024}
}

@article{donnelly2025liquidity,
  title={Liquidity competition between brokers and an informed trader},
  author={Donnelly, Ryan and Li, Zi},
  journal={arXiv preprint arXiv:2503.08287},
  year={2025}
}

@article{cartea2024nash,
title={Nash equilibrium between brokers and traders},
author={Cartea, {\'A}lvaro and Jaimungal, Sebastian and S{\'a}nchez-Betancourt, Leandro},
journal={Finance and Stochastics, to appear, arXiv:2407.10561},
year={2025+}
}

@article{fawcett-roc,
title = {An introduction to ROC analysis},
journal = {Pattern Recognition Letters},
volume = {27},
number = {8},
pages = {861-874},
year = {2006},
note = {ROC Analysis in Pattern Recognition},
issn = {0167-8655},
doi = {https://doi.org/10.1016/j.patrec.2005.10.010},
url = {https://www.sciencedirect.com/science/article/pii/S016786550500303X},
author = {Tom Fawcett}
}

@article{cartea2023optimal,
  title={Optimal execution with stochastic delay},
  author={Cartea, {\'A}lvaro and S{\'a}nchez-Betancourt, Leandro},
  journal={Finance and Stochastics},
  volume={27},
  number={1},
  pages={1--47},
  year={2023},
  publisher={Springer}
}

@article{theonlygameintown,
author = {Walter Bagehot},
title = {The Only Game in Town},
journal = {Financial Analysts Journal},
volume = {27},
number = {2},
pages = {12-14},
year  = {1971},
publisher = {Routledge},
doi = {10.2469/faj.v27.n2.12},
URL = { https://doi.org/10.2469/faj.v27.n2.12},
eprint = {  https://doi.org/10.2469/faj.v27.n2.12}
}

@article{easley1996liquidity,
  title={Liquidity, information, and infrequently traded stocks},
  author={Easley, David and Kiefer, Nicholas M and O'{H}ara, Maureen and Paperman, Joseph B},
  journal={The Journal of Finance},
  volume={51},
  number={4},
  pages={1405--1436},
  year={1996},
  publisher={Wiley Online Library}
}

@misc{frankle2019lottery,
      title={The Lottery Ticket Hypothesis: Finding Sparse, Trainable Neural Networks}, 
      author={Jonathan Frankle and Michael Carbin},
      year={2019},
      eprint={1803.03635},
      archivePrefix={arXiv},
      primaryClass={cs.LG}
}

@misc{li2018subspacenn,
      title={Measuring the Intrinsic Dimension of Objective Landscapes}, 
      author={Chunyuan Li and Heerad Farkhoor and Rosanne Liu and Jason Yosinski},
      year={2018},
      eprint={1804.08838},
      archivePrefix={arXiv},
      primaryClass={cs.LG}
}

@article{copelandgalai,
 ISSN = {00221082, 15406261},
 URL = {http://www.jstor.org/stable/2327580},
 author = {Thomas E. Copeland and Dan Galai},
 journal = {The Journal of Finance},
 number = {5},
 pages = {1457--1469},
 publisher = {[American Finance Association, Wiley]},
 title = {Information Effects on the Bid-Ask Spread},
 urldate = {2023-07-05},
 volume = {38},
 year = {1983}
}

@article{RoelAggregator,
author = {Roel Oomen},
title = {Execution in an aggregator},
journal = {Quantitative Finance},
volume = {17},
number = {3},
pages = {383-404},
year  = {2017},
publisher = {Routledge},
doi = {10.1080/14697688.2016.1201589},

URL = {https://doi.org/10.1080/14697688.2016.1201589},
eprint = { 
    https://doi.org/10.1080/14697688.2016.1201589
        }
}

@article{glosten1985bid,
  title={Bid, ask and transaction prices in a specialist market with heterogeneously informed traders},
  author={Glosten, Lawrence R and Milgrom, Paul R},
  journal={Journal of Financial Economics},
  volume={14},
  number={1},
  pages={71--100},
  year={1985},
  publisher={Elsevier}
}

@article{amihud1980dealership,
  title={Dealership market: Market-making with inventory},
  author={Amihud, Yakov and Mendelson, Haim},
  journal={Journal of Financial Economics},
  volume={8},
  number={1},
  pages={31--53},
  year={1980},
  publisher={Elsevier}
}

@Article{lambert2021-rvga,
author={Lambert, Marc
and Bonnabel, Silv{\`e}re
and Bach, Francis},
title={The recursive variational {G}aussian approximation ({R-VGA})},
journal={Statistics and Computing},
year={2021},
month={Dec},
day={20},
volume={32},
number={1},
pages={10},
issn={1573-1375},
doi={10.1007/s11222-021-10068-w},
url={https://doi.org/10.1007/s11222-021-10068-w}
}

@misc{ollivier17-ekf-natural-gradient,
  doi = {10.48550/ARXIV.1703.00209},
  url = {https://arxiv.org/abs/1703.00209},
  author = {Ollivier, Yann},
  title = {Online Natural Gradient as a {K}alman Filter},
  publisher = {arXiv},
  year = {2017},  
  copyright = {arXiv.org perpetual, non-exclusive license}
}

@inproceedings{duran-martin22-subspace-ekf,
  title={Efficient Online Bayesian Inference for Neural Bandits},
  author={Duran-Martin, Gerardo and Kara, Aleyna and Murphy, Kevin},
  booktitle={International Conference on Artificial Intelligence and Statistics},
  pages={6002--6021},
  year={2022},
  organization={PMLR}
}

@book{pml1Book,
 author = "Kevin P. Murphy",
 title = "Probabilistic Machine Learning: An Introduction",
 publisher = "MIT Press",
 year = 2022,
 url = "probml.ai"
}

@book{pml2Book,
 author = "Kevin P. Murphy",
 title = "Probabilistic Machine Learning: Advanced Topics",
 publisher = "MIT Press",
 year = 2023,
 url = "http://probml.github.io/book2"
}

@misc{dof-dnns-2021,
  doi = {10.48550/ARXIV.2107.05802},
  url = {https://arxiv.org/abs/2107.05802},
  author = {Larsen, Brett W. and Fort, Stanislav and Becker, Nic and Ganguli, Surya},
  title = {How many degrees of freedom do we need to train deep networks: a loss landscape perspective},
  publisher = {arXiv},
  year = {2021},
  copyright = {arXiv.org perpetual, non-exclusive license}
}

@inproceedings{kingma14-adam,
  author    = {Diederik P. Kingma and
               Jimmy Ba},
  editor    = {Yoshua Bengio and
               Yann LeCun},
  title     = {Adam: {A} Method for Stochastic Optimization},
  booktitle = {3rd International Conference on Learning Representations, {ICLR} 2015},
  year      = {2015}
 }

@book{hastie01-esl,
  address = {New York, NY, USA},
  author = {Hastie, Trevor and Tibshirani, Robert and Friedman, Jerome},
  publisher = {Springer New York Inc.},
  series = {Springer Series in Statistics},
  title = {The Elements of Statistical Learning},
  year = 2001
}

@article{liu89-lbfgs,
  doi = {10.1007/bf01589116},
  url = {https://doi.org/10.1007/bf01589116},
  year = {1989},
  month = aug,
  publisher = {Springer Science and Business Media {LLC}},
  volume = {45},
  number = {1-3},
  pages = {503--528},
  author = {Dong C. Liu and Jorge Nocedal},
  title = {On the limited memory {BFGS} method for large scale optimization},
  journal = {Mathematical Programming}
}

@misc{lin2019steins,
      title={Stein's Lemma for the Reparameterization Trick with Exponential Family Mixtures}, 
      author={Wu Lin and Mohammad Emtiyaz Khan and Mark Schmidt},
      year={2019},
      eprint={1910.13398},
      archivePrefix={arXiv},
      primaryClass={stat.ML}
}

@techreport{ait2022and,
  title={How and When are High-Frequency Stock Returns Predictable?},
  author={A{\"\i}t-Sahalia, Yacine and Fan, Jianqing and Xue, Lirong and Zhou, Yifeng},
  year={2022},
  institution={National Bureau of Economic Research}
}

@article{cartea2022brokers,
  author={Cartea, {\'A}lvaro and S{\'a}nchez-Betancourt, Leandro},
  title={Brokers and informed traders: dealing with toxic flow and extracting trading signals},
  journal={SIAM Journal on Financial Mathematics},
  volume={16},
  number={2},
  pages={243--270},
  year={2025},
  publisher={SIAM}
}

@article{grossman1980impossibility,
  title={On the impossibility of informationally efficient markets},
  author={Grossman, Sanford J and Stiglitz, Joseph E},
  journal={The American Economic Review},
  volume={70},
  number={3},
  pages={393--408},
  year={1980},
  publisher={JSTOR}
}

@article{kyle1985continuous,
  title={Continuous auctions and insider trading},
  author={Kyle, Albert S},
  journal={Econometrica: Journal of the Econometric Society},
  pages={1315--1335},
  year={1985},
  publisher={JSTOR}
}

@article{butz2019internalisation,
  title={Internalisation by electronic {FX} spot dealers},
  author={Butz, Maximilian and Oomen, Roel},
  journal={Quantitative Finance},
  volume={19},
  number={1},
  pages={35--56},
  year={2019},
  publisher={Taylor \& Francis}
}

@article{kyle1989informed,
  title={Informed speculation with imperfect competition},
  author={Kyle, Albert S},
  journal={The Review of Economic Studies},
  volume={56},
  number={3},
  pages={317--355},
  year={1989},
  publisher={Wiley-Blackwell}
}

@article{arroyo2023deep,
  title={Deep attentive survival analysis in limit order books: Estimating fill probabilities with convolutional-transformers},
  author={Arroyo, Alvaro and Cartea, Alvaro and Moreno-Pino, Fernando and Zohren, Stefan},
  journal={Quantitative Finance},
  volume={24},
  number={1},
  pages={35--57},
  year={2024},
  publisher={Taylor \& Francis}
}

@article{cartea2023bandits,
  title={Bandits for Algorithmic Trading with Signals},
  author={Cartea, {\'A}lvaro and Drissi, Fay{\c{c}}al and Osselin, Pierre},
  journal={Available at SSRN 4484004},
  year={2023}
}

\appendix

\section{Features}\label{app: features}
For a given trading day ${\mathfrak d}\in \mfD$, the processes
\begin{equation*}
    \left(S^{a,\mathfrak{d}}_t\right)_{t\in\mfT}\,,\quad
    \left(S^{b,\mathfrak{d}}_t\right)_{t\in\mfT}\,,\quad 
    \left(V^{a,\mathfrak{d}}_t\right)_{t\in\mfT}\,,\quad 
    \left(V^{b,\mathfrak{d}}_t\right)_{t\in\mfT}\,,
\end{equation*}
denote the best ask price, the best bid price, the volume  at the best ask price, 
and the volume  at the best bid price in LMAX Exchange for day $\mathfrak{d}$, respectively ---
we drop the superscript $\mathfrak{d}$ when we do not wish to draw attention to the day.
The feature associated with 
the log-transformed inventory of client $c\in\mcA$ is
\begin{equation*}
\text{sign}\left(\inventory^{c}_{t^-}\right) \times \log\left(1+|\inventory^{c}_{t^-}|\right)\,,
\end{equation*}
where $\inventory^c_{t^-}$ is the position in lots (one lot is \euro 10,000)  of client $c$ accumulated
over $[0,t)$. 
Here, $q^c_t$ is the size of the order sent at time $t$ by client $c$
and $N^{c,a}_t$, $N^{c,b}_t $ are the counting processes of buy and sell orders, respectively, sent by client $c$
and filled by the broker.
The cash of client $c$, denoted by $\cash^{c}_t$, is given by
\begin{equation*}
\cash^{c}_t = - \int_0^t S^a_{u^-}\,q^c_u\,\diff N^{a,c}_u +  \int_0^t  S^b_{u^-}\,q^c_u\,\diff N^{b,c}_u\,, 
\end{equation*}
and the feature associated with the cash process is 
\begin{equation*}
\text{sign}\left(\cash^{c}_{t^-}\right) \times \log\left(1+|\cash^{c}_{t^-}|\right)\,.
\end{equation*}
In LMAX Exchange, the bid-ask spread is
\begin{equation*}
    \spread_t = S^a_t - S^b_t \,,
\end{equation*}
the midprice is 
\begin{equation*}
    S_t = \frac{S^a_t + S^b_t}{2} \,,
\end{equation*}
the volume imbalance of the best available volumes is defined by
\begin{equation*}
    \imbalance_t = \frac{V_t^\text{b} - V_t^\text{a}}{V_t^\text{b} + V_t^\text{a}}\,,
\end{equation*}
and the associated feature for the volume $V$ is the transformed volume
\begin{equation*}
    \volatility = \log(1+|V|)\,.
\end{equation*}
The number of trades received by the broker from her clients is
\begin{equation*}
    N_t = \sum_{c\in\mcA} \left(N^{c,a}_t+ N^{c,b}_t\right)\,,
\end{equation*}
and the volatility of the midprice in the LOB of LMAX Exchange over the interval $[t-\delta,t)$ is the square root of the quadratic variation of the logarithm of the midprice over the interval. More precisely, 
\begin{equation*}
    \volatility^{\delta}_t = \sqrt{ \sum_{\Delta \log{ S_u} \neq 0 \,;\, u\in [t-\delta,t)} |{\Delta \log{ S_u}}|^2 }\,,
\end{equation*}
where 
\begin{equation*}
  \Delta \log{ S_u} = \log{ S_u} - \log{ S_{u^-}} \,,\text{ and } S_{u^-} = \lim_{v\nearrow u} { S_v}\,.
\end{equation*}
The return of the exchange rate of the currency pair over a period $\delta>0$ is given by
\begin{equation*}
   \return_t^\delta = \log\left({S_{t^-}}/{S_{t-\delta}}\right).
\end{equation*}

The timing of the events in the LOB is measured with three clocks:
time-clock,  transaction-clock, and volume-clock.
The time-clock runs as $t\in[0,T]$ with microsecond precision,
i.e., a millionth of a second.
For a given day $d$ with $N^d$ transactions and $V^d$ volume traded, the transaction clock runs as $\mathfrak{t}\in[0,N^d]$, and the volume-clock runs as $\mathfrak{v}\in[0,V^d]$.  
The number of transactions $\mathfrak{t}(t)$ is the number of transactions observed up until time $t$, similar for the volume-clock $\mathfrak{v}(t)$. Thus, for any order sent at time $t$, the time associated with the order in the transaction-clock is $\mathfrak{t}(t)$ and the time in the volume-clock is $\mathfrak{v}(t)$.

For each clock $\clock \in\{$transaction, time, volume$\}$,
we build features spanning an interval $[\window_\clock\,2^n, \window_\clock\,2^{n+1})$
of units in the respective clock with $\window_\clock>0$,
and use a given statistic to summarise the values in the interval;
for example, for spread, imbalance, and transformed volumes we use
the average value over the period.
In our experiments, we consider seven intervals that span the ranges $[0,\window_\clock)$ and $\{[\window_\clock\,2^n,\, \window_\clock\,2^{n+1})\}_{n=0}^{8}$. 
The median time elapsed between any two transactions for the six clients is 1.8 seconds and the median quantity traded with LMAX Broker is \euro2,000. 
Thus, $\window_{\text{transaction}}$ is one transaction,
$\window_{\text{time}}$ is one second, and
 $\window_{\text{volume}}$ is \euro 2,000.\footnote{The intervals for the transaction-clock are $[0,1)$ transaction, $[1,2)$ transactions, $[2,4)$ transactions, \dots, and $[2^{5}, 2^{6})$ transactions. 
Similarly, the intervals  for the time-clock are $[0,1)$ second, $[1,2)$ seconds, $[2,4)$ seconds, \dots, and $[2^{5}, 2^{6})$ seconds. 
Lastly, the intervals  in the volume-clock are
$[0,\text{\euro} 2000)$, $[\text{\euro} 2000, \text{\euro} 4000)$,
$[\text{\euro} 4000, \text{\euro} 8000)$, \dots,
and $[\text{\euro} 2000\times 2^{5}, \text{\euro} 2000\times 2^{6})$. }

\section{{\ourmodel} derivation and proofs}\label{app: ourmodel derivations}
In this section, we present the  derivation and proofs for {\ourmodel}.
Proposition \ref{app: prop: fixed point} derives the general fixed point equations.
Given that these are expensive to compute, we present additional results to obtain the computationally efficient form of the theorem.
We then prove Theorem \ref{theorem:subspace-last-rvga} in  \ref{prop:order-2-single-step}.

\begin{proposition}\label{app: prop: fixed point}
(Modified R-VGA for {\ourmodel}) Suppose $\log p(\yt[n] \,\vert\, \phiddensub, \plast; \xt[n])$ is differentiable with respect to
to $(\phiddensub, \plast)$ and the observations $\{\yt[n]\}_{n=1}^N$ are conditionally independent over
$(\phiddensub, \plast)$. Given Gaussian prior distributions $\phi_0$, $\varphi_0$ for $\plast$ and $\phiddensub$
respectively, the variational posterior distributions at time $n \in \{1, \ldots, N\}$  that solve  \eqref{eq:subspace-last-rvga}
satisfy the fixed-point equations
\begin{equation}
\begin{aligned}\label{eq:ourmodel-form1}
    \mlast_n &=
    \mlast_{n-1} - 
    \covlast_{n-1} \nabla_{\mlast}\,\E_{\vdlast[n], \vdhidden[n]}[\log p(y_n \,\vert\, \phiddensub, \plast; \xt[n])]\,,\\
    \mhidden_n &= \mhidden_{n-1} - \covhidden_{n-1} \nabla_{\mhidden}\,\E_{\vdhidden[n], \vdhidden[n]}[\log p(\yt[n] \,\vert\, \phiddensub, \phidden; \xt[n])]\,,\\
    \covlast_n^{-1} &= \covlast_{n-1}^{-1} + 2\,\nabla_\covlast \,\E_{\vdlast[n], \vdhidden[n]}[\log p(\yt[n] \,\vert\, \phiddensub, \plast; \xt[n])]\,,\\
    \covhidden_n^{-1} &= \covhidden_{n-1}^{-1} + 2\,\nabla_\covhidden\, \E_{\vdhidden[n], \vdhidden[n]}[\log p(\yt[n] \,\vert\, \phiddensub, \phidden; \xt[n])]\,.
\end{aligned}
\end{equation}
\end{proposition}

\begin{proof}
First,  rewrite \eqref{eq:subspace-last-rvga}. Let $p(y_n) \equiv p(\yt[n] \,\vert\, {\bf z}, \plast; \xt[n])$
to simplify notation. The loss function is
\begin{align}
    {\cal K}_n &= \text{KL}\left(
    \normdistlast\,\normdisthidden \, ||\, \phi_{n-1}(\plast)\,\varphi_{n-1}({\bf z})\, p(\yt[n])
    \right)\nonumber\\
    &= \iint \normdisthidden\,\normdistlast \log\left( \frac{\normdisthidden\,\normdistlast}
        { \vdhidden[n-1]\, \vdlast[n-1]\, p(\yt[n])} \right)
        \diff\phiddensub\,\diff\plast\nonumber\\
    &= \iint \normdisthidden\,\normdistlast \left[
        \log\left(\frac{\normdisthidden}{\vdhidden[n-1]}\right) + \log\left(\frac{\normdisthidden}{\vdlast[n-1]} \right) - \log p(y_n)
    \right] \diff\phiddensub\, \diff\plast\nonumber\\
    &= \int \normdisthidden\, \log\left(\frac{\normdisthidden}{\vdhidden[n-1]}\right) \diff\phiddensub + \int \normdistlast \log\left(\frac{\normdistlast}{\vdlast[n-1]}\right) \diff\plast \nonumber \nonumber\\ 
    &\quad  + \iint \normdisthidden\,\normdistlast\, \log p(y_n)\, \diff\plast \,\diff\phiddensub \nonumber\\
    &= \mathbb{E}_{\normdisthidden}\left[
        \log\left(\frac{\normdisthidden}{\vdhidden[n-1]}\right)
    \right] + \mathbb{E}_{\normdistlast}\left[
        \log\left(\frac{\normdistlast}{\vdlast[n-1]}\right)
    \right]
\nonumber \\
    &\quad    + \mathbb{E}_{\normdisthidden\, \normdistlast}[\log p(y_n)] \nonumber\\
    &\begin{aligned}\label{eq:rvga-fsll-rewrite}
    &=\text{KL}\left(\normdistlast \, ||\, \vdlast[n-1]\right)
    + \text{KL}\left(\normdisthidden \, || \, \vdhidden[n-1]\right)\\
    &\hspace{7cm} + \mathbb{E}_{\normdisthidden\, \normdistlast}[\log p(y_n)]\,.
    \end{aligned}
\end{align}
The first and second terms in \eqref{eq:rvga-fsll-rewrite} correspond to a Kullback--Leibler
divergence between two multivariate Gaussians.
The last term corresponds to the posterior-predictive marginal log-likelihood for the $t$-th observation.
To minimise \eqref{eq:rvga-fsll-rewrite}, we recall that the Kullback--Leibler divergence between two multivariate Gaussian
is given by
\begin{align*}
    &\text{KL}(\normdist{\bf x}{{\bf m}_1}{{\bf S}_1}\, || \, \normdist{\bf x}{{\bf m}_2}{{\bf S}_2}) \\
    &= \frac{1}{2}\left[
        \text{Tr}\left( {\bf S}_2^{-1}{\bf S}_1 \right)
        + ({\bf m}_2 - {\bf m}_1)^\intercal {\bf S}_2^{-1} ({\bf m}_2 - {\bf m}_1)
        - M
        + \log\left( |{\bf S}_2|/|{\bf S}_1|\right)
    \right]\,,
\end{align*}
see Section 6.2.3 in \cite{pml1Book}.

To simplify notation, let $\mathbb{E}_{\normdisthidden \normdistlast}[\log p(y_n)] =: {\cal E}_n$.
The derivative of ${\cal K}_n$ with respect to $\mlast$ is 
\begin{align}
    \nabla_\mlast {\cal K}_n
    &= \nabla_\mlast \left(
    \text{KL}\left(\normdistlast \, ||\, \vdlast[n-1]\right)
    + {\cal E}_n
    \right) \nonumber\\ 
    &= \nabla_\mlast\left(
        \frac{1}{2}\mlast^\intercal\covlast_{n-1}^{-1}\mlast - \mlast^\intercal\covlast_{n-1}\mlast_{n-1}
         + \nabla_\mlast{\cal E}_n
    \right)\nonumber \\
    &= \covlast_{n-1}^{-1}\mlast - \covlast_{n-1}^{-1}\mlast_{n-1} +\nabla_\mlast {\cal E}_n \nonumber \\
    &= \covlast_{n-1}^{-1}\left( \mlast - \mlast_{n-1} - \covlast_{n-1}\nabla_\mlast{\cal E}_n \right). \label{eq:dK-mlast}
\end{align}
Set \eqref{eq:dK-mlast} to zero and solve for
\begin{equation*}
    \mlast = \mlast_{n-1} - \covlast_{n-1}\nabla_\mlast{\cal E}_n\,.
\end{equation*}
Next, we estimate the condition
for $\covlast$. Use \eqref{eq:rvga-fsll-rewrite} to obtain
\begin{align}
    \nabla_\covlast {\cal K}_n &=
    \nabla_\covlast\left(
        -\frac{1}{2}\,\log|\covlast| + \frac{1}{2}\,\text{Tr}\left(\covlast\,\covlast_{n-1}^{-1}\right) + {\cal E}_n
    \right)\nonumber \\
    &= -\frac{1}{2}\,\covlast^{-1} + \frac{1}{2}\,\covlast_{n-1}^{-1} + \nabla_{\covlast}{\cal E}_n. \label{eq:dK-covlast}
\end{align}
The fixed-point solution for \eqref{eq:dK-covlast} satisfies
\begin{equation*}
    \covlast^{-1} = \covlast_{n-1}^{-1} + 2\, \nabla_{\covlast}{\cal E}_n\,.
\end{equation*}
We derive the fixed-point conditions for $\mhidden$ and $\covhidden$ similarly.
\end{proof}

\begin{corollary}\label{cor:rvga-order-2}
Suppose $\log p(y \,\vert\, \phiddensub, \plast)$ is differentiable with respect
to $(\phiddensub, \plast)$ and the observations $\{y_n\}_{n=1}^T$ are conditionally independent over
$(\phiddensub, \plast)$. Given Gaussian prior distributions $\phi_0$, $\varphi_0$ for $\plast$ and $\phiddensub$
respectively, the modified R-VGA equations for {\ourmodel}
in terms of gradients and Hessians with respect to  $\phiddensub$ and $\plast$ are
\begin{align*}
    \mlast_n &=
    \mlast_{n-1} + \covlast_{n-1} \E_{\vdlast[n]\,\vdhidden[n]}[\nabla_\plast \log p(y_n \,\vert\, \phiddensub, \plast; \xt[n])]\,,\\
    \mhidden_n &=
    \mhidden_{n-1} + \covhidden_{n-1} \E_{\vdlast[n]\,\vdhidden[n]}[\nabla_\phiddensub \log p(y_n \,\vert\, \phiddensub, \plast; \xt[n])]\,,\\
    \covlast_n^{-1} &=
    \covlast_{n-1}^{-1} + \E_{\vdlast[n]\,\vdhidden[n]}[\nabla_\plast^2 \log p(y_n \,\vert\, \phiddensub, \plast; \xt[n])]\,,\\
    \covhidden_n^{-1} &=
    \covhidden_{n-1}^{-1} + \E_{\vdlast[n]\,\vdhidden[n]}[\nabla_\phiddensub^2 \log p(y_n \,\vert\, \phiddensub, \plast; \xt[n])]\,.
\end{align*}
\end{corollary}
\begin{proof}
    The proof follows directly by rearranging the order of integration, Bonnet’s Theorem, and Prices’s Theorem.
    see Theorem 3 and Theorem 4 in \cite{lin2019steins}.
\end{proof}

Corollary \ref{cor:rvga-order-2} provides tractable fixed-point equations. Note that in Proposition \ref{app: prop: fixed point} the gradients are taken with respect to model parameters outside the expectation, whereas in Corollary \ref{cor:rvga-order-2} the gradients and Hessians are taken inside the expectation.

\begin{proposition}\label{prop:part-derivative-canonical-link-bernoulli}
Given a logistic regression model written in the canonical form
\begin{align*}
    \log p(y_n) &=
    y_n \log\left(\frac{\sigma({\bm\theta}^\intercal\,\xt[n])}{1 -\sigma(\bm{\theta}^\intercal\,\xt[n]) }\right) -
    \left(1 + \exp(\sigma(\bm{\theta}^\intercal\,\xt[n]))\right)\,,
\end{align*}
we have that the gradient of $\log\left(\frac{\sigma({\bm\theta}^\intercal\,\xt[n])}{1 -\sigma(\bm{\theta}^\intercal\,\xt[n]) }\right)$
with respect to
$\bm\theta$ is given by
\begin{equation*}
    \nabla_{\bm{\theta}} \log\left(\frac{\sigma({\bm\theta}^\intercal\,\xt[n])}{1 -\sigma(\bm{\theta}^\intercal\,\xt[n]) }\right) = \xt[n]\,.
\end{equation*}
\end{proposition}
\begin{proof}
    Use the identity $\frac{\diff}{\diff x}\sigma(x) = \sigma(x) \,(1 - \sigma(x)) =: \sigma'(x)$ and write
    \begin{align*}
        \nabla_{\bm\theta} \log\left(\frac{\sigma({\bm\theta}^\intercal\,\xt[n])}{1 -\sigma(\bm{\theta}^\intercal\,\xt[n]) }\right)
        &= \nabla_{\bm\theta} \log\left(\sigma({\bm\theta}^\intercal\,\xt[n]\right) - \nabla_{\bm\theta}\log\left(1 -\sigma(\bm{\theta}^\intercal\,\xt[n])\right)\\
        &= \left(\sigma({\bm\theta}^\intercal\,\xt[n])\right)^{-1}\sigma'({\bm\theta}^\intercal\,\xt[n])\,\xt[n]
        + \left(1 -\sigma(\bm{\theta}^\intercal\,\xt[n])\right)^{-1}\sigma'({\bm\theta}^\intercal\,\xt[n])\xt[n]\\
        &= (1 - \sigma(\bm{\theta}^\intercal\,\xt[n]) + \sigma(\bm{\theta}^\intercal\,\xt[n]))\,\xt[n]\\
        &= \xt[n]\,.
    \end{align*}
\end{proof}

\subsection{Proof of Theorem \ref{theorem:subspace-last-rvga}}\label{prop:order-2-single-step}
\begin{proof}
    First rewrite the log-likelihood of the target variable $y_n$ as a member of the exponential-family.
    Let $f_n(\phiddensub, \plast) = \plast^\intercal\, h(\phiddensub; \xt[n])$. Then,
    \begin{align*}
        \log p(y_n)
        &= \log \text{Bern}\Big(\yt[n] \,\vert\, \sigma(\plast^\intercal\, h(\phiddensub; \xt[n]))\Big)\\
        &= \yt[n] \log \sigma(f_n(\phiddensub, \plast)) + (1 - \yt[n]) \log(1 - \sigma(f_n(\phiddensub, \plast)))\\
        &= \yt[n] \log\left( \frac{\sigma(f_n(\phiddensub, \plast))}{1 - \sigma(f_n(\phiddensub, \plast))} \right)
             + \log(1 - \sigma(f_n(\phiddensub, \plast)))\\
        &= \yt[n] \log\left( \frac{\sigma(f_n(\phiddensub, \plast))}{1 - \sigma(f_n(\phiddensub, \plast))} \right)
             - \log\left(1 + \exp\left(\log\left(\frac{\sigma(f_n(\phiddensub, \plast))}{1 - \sigma(f_n(\phiddensub, \plast))}\right)\right)\right)\\
        &= \yt[n]\,\eta_n - \log(1 + \exp(\eta_n))\\
        &= \yt[n]\,\eta_n - A(\eta_n)\,,
    \end{align*}
    where $\eta_n = \log\left( \frac{\sigma(f_n(\phiddensub, \plast))}{1 - \sigma(f_n(\phiddensub, \plast))} \right)$ is
    the natural parameter and $A(\eta_n)$ the log-partition function.
    We perform a moment-match estimation. 
    To do this, we follow the property of exponential-family distributions,
    and use the results in Proposition \ref{prop:part-derivative-canonical-link-bernoulli}.
    The first and second-order derivatives of the log-partition $A(\eta_n)$ are
    \begin{align}
        \frac{\partial}{\partial \eta_n} A(\eta_n) &= \E[y \,\vert\, \eta_n]
        = \sigma(f_n(\phiddensub, \plast))\,, \label{eq:part-model-expectation}\\
        \frac{\partial^2}{\partial \eta_n^2} A(\eta_n) &= \text{Cov}[y \,\vert\, \eta_n]
        = \sigma'(f_n(\phiddensub, \plast))\,. \label{eq:part-model-covariance}
    \end{align}
    Then, the first order approximation 
    $\hat\sigma(f_n(\phiddensub, \plast))$ is
    \begin{align*}
        \hat\sigma(f_n(\phiddensub, \plast))
        &= \sigma(\bar{f}_{n-1}) + \sigma'(\bar{f}_{n-1}) (F_{n,\phiddensub}^\intercal\,(\phiddensub - \mhidden_{n-1}) +  F_{n,\plast}^\intercal\,(\plast - \mlast_{n-1}))\,,
    \end{align*}
    where $\bar{f}_{n-1} = f_n(\mhidden_{n-1}, \mlast_{n-1})$,
    $F_{n,\phiddensub} = \nabla_{\phiddensub}f_n(\phiddensub, \mlast_{n-1})\big\vert_{\phiddensub=\mhidden_{n-1}}$, and 
    $F_{n,\plast} = \nabla_{\plast}f_n(\mhidden_{n-1}, \plast)\big\vert_{\plast=\mlast_{n-1}} =h(\mhidden_{n-1}; \xt[n])$.
    Next, we derive the update equations for $\mlast$ and $\covlast$. The  moment-matched log-likelihood is given by 
    \begin{equation}
        \log p(y_n) = \log\normdist{y_n}
        {\hat\sigma(f_n(\phiddensub, \plast))}{\sigma(f_n(\phiddensub, \plast)) ( 1 - \sigma(f_n(\phiddensub, \plast)))}
    \end{equation}
    and the gradient with respect to $\plast$ is
    \begin{equation*}
        \nabla_\plast \log p(y_n) = \yt[n]\, F_{n,\plast} - \hat\sigma(\bar{f}_{n-1})\, F_{n,\plast}
        = (\yt[n] - \hat\sigma(\bar{f}_{n-1}))\, F_{n,\plast}\,.
    \end{equation*}
    Now, the Hessian of the log-model with respect to $\plast$ is
    \begin{align*}
        \nabla_\plast^2 \log p(y_n) &= \nabla_\plast (\yt[n] - \hat\sigma(\bar{f}_{n-1}))\, F_{n,\plast}\\
        &= -\sigma'(\bar{f}_{n-1})\, F_{n,\plast}^\intercal\, \,F_{n,\plast}\,.
    \end{align*}
    The update for $\covlast_n$, following the order-2 form of the modified equations and replacing
    the expectation under $\vdhidden[n]$ for $\vdhidden[n-1]$, is
    \begin{align*}
        \covlast_n^{-1}
        &= \covlast_{n-1}^{-1} - \E_{\vdlast[n]\, \vdhidden[n-1]}\left[\nabla^2_\plast \log p(y_n)\right]\\
        &= \covlast_{n-1}^{-1} - \E_{\vdlast[n]\, \vdhidden[n-1]}\left[ -\sigma'(\bar{f}_{n-1})\, F_{n,\plast}^\intercal\, F_{n,\plast} \right]\\
        &= \covlast_{n-1}^{-1} + \sigma'(\bar{f}_{n-1})\, F_{n,\plast}^\intercal\, F_{n,\plast}\\
        &= \covlast_{n-1}^{-1} + \sigma'\big(\mlast_{n-1}^\intercal\, h(\mhidden_{n-1}; \xt[n])\big)\, h(\mhidden_{n-1}; \xt[n])^\intercal\, h(\mhidden_{n-1}; \xt[n])\,.
    \end{align*}
    Next, the update step for $\mlast_n$ becomes
    \begin{align*}
        \mlast_n &= \mlast_{n-1} - \covlast_{n-1}\, \E_{\vdlast[n] \vdhidden[n-1]}[\nabla_\plast \log p(y_n \,\vert\, \phiddensub, \plast; \xt[n])]\\
        &= \mlast_{n-1} + \covlast_{n-1}\, \E_{\vdlast[n] \vdhidden[n-1]}[(\yt[n] - \hat\sigma(\bar{f}_{n-1}))\, F_{n,\phiddensub}]\\
        &= \mlast_{n-1} + \covlast_{n-1}\, \E_{\vdlast[n] \vdhidden[n-1]}\Big[(\yt[n] - 
        \sigma(\bar{f}_{n-1})  \\
        &\qquad + \sigma'(\bar{f}_{n-1})
        (F_{n,\phiddensub}(\phiddensub - \mhidden_{n-1}) + F_{n,\plast}^\intercal\,(\plast - \mlast_{n-1}))\Big]F_{n,\phiddensub} \nonumber\\
        &= \mlast_{n-1} + \covlast_{n-1}\, F_{n, \plast}\Big(\yt[n] - 
        \sigma(\bar{f}_{n-1})
        - \sigma'(\bar{f}_{n-1})\, F_{n,\plast}^\intercal\,(\mlast_n - \mlast_{n-1})\Big)\\
        &= \mlast_{n-1} + \covlast_{n-1}\, F_{n, \plast}\Big(\yt[n] - 
        \sigma(\bar{f}_{n-1})
        + \sigma'(\bar{f}_{n-1})\, F_{n,\plast}^\intercal\,\mlast_{n-1})\Big)
        - \sigma'(\bar{f}_{n-1})\covlast_{n-1}\, F_{n, \plast}\, F_{n, \plast}^\intercal\,\mlast_n\,.
    \end{align*}
We rewrite the last equality as
\begin{align*}
    &\mlast_n + \sigma'(\bar{f}_{n-1})\covlast_{n-1}\, F_{n, \plast}\, F_{n, \plast}^\intercal\,\mlast_n
    = \mlast_{n-1} + \covlast_{n-1}\, F_{n, \plast}\Big(\yt[n] - 
    \sigma(\bar{f}_{n-1})
    + \sigma'(\bar{f}_{n-1})\, F_{n,\plast}^\intercal\,\mlast_{n-1})\Big)\,,
\end{align*}
which implies that
\begin{align*}
   \left({\bf I} + \sigma'(\bar{f}_{n-1})\,\covlast_{n-1}\,F_{n, \plast}\,F_{n, \plast}^\intercal\,\right)\mlast_n
    = \mlast_{n-1} + \covlast_{n-1}\,F_{n, \plast}\,\Big(\yt[n] - 
    \sigma(\bar{f}_{n-1})
    + \sigma'(\bar{f}_{n-1})\,F_{n,\plast}^\intercal\,\mlast_{n-1})\Big)\,.
\end{align*}
Similarly, we have that
\begin{align*}
   \mlast_n
    = \left({\bf I} + \sigma'(\bar{f}_{n-1})\covlast_{n-1}\, F_{n, \plast}F_{n, \plast}^\intercal\,\right)^{-1}
    \left(\mlast_{n-1} + \covlast_{n-1}F_{n, \plast}\Big(\yt[n] - 
    \sigma(\bar{f}_{n-1})
    + \sigma'(\bar{f}_{n-1})\, F_{n,\plast}^\intercal\,\mlast_{n-1})\Big)\right)\,,
\end{align*}
where 
\begin{align*}
    \left({\bf I} + \sigma'(\bar{f}_{n-1})\,\covlast_{n-1}\,F_{n, \plast}F_{n, \plast}^\intercal\,\right)^{-1}
    =&\left(\covlast_{n-1}\left[\covlast_{n-1}^{-1} + \sigma'(\bar{f}_{n-1})F_{n, \plast}F_{n, \plast}^\intercal\,\right]\right)^{-1}\\
    =& \left[\covlast_{n-1}^{-1} + \sigma'(\bar{f}_{n-1})F_{n, \plast}F_{n, \plast}^\intercal\,\right]^{-1}\covlast_{n-1}^{-1}\\
    =&\, \covlast_{n}\,\covlast_{n-1}^{-1}\,.
\end{align*}
    Then, it follows that
    \begin{align*}
        \mlast_n
        &= \covlast_{n}\,\covlast_{n-1}^{-1}
        \left(\mlast_{n-1} + \covlast_{n-1}\,F_{n, \plast}\Big(\yt[n] - 
        \sigma(\bar{f}_{n-1})
        + \sigma'(\bar{f}_{n-1})\,F_{n,\plast}^\intercal\,\mlast_{n-1})\Big)\right)\\
        &= 
        \covlast_{n}\,\covlast_{n-1}^{-1}\,\mlast_{n-1} +
        \covlast_{n}\,\covlast_{n-1}^{-1}\,\covlast_{n-1}\,F_{n, \plast}\Big(\yt[n] - 
        \sigma(\bar{f}_{n-1})
        + \sigma'(\bar{f}_{n-1})\,F_{n,\plast}^\intercal\,\mlast_{n-1})\Big)\\
        &=
        \covlast_{n}\,\left(\covlast_{n-1}^{-1} + \sigma'(\bar{f}_{n-1})\,F_{n,\plast}^\intercal\,\mlast_{n-1}) \right)\, \mlast_{n-1}
        +\covlast_{n-1}\,F_{n, \plast}\Big(\yt[n] - \sigma(\bar{f}_{n-1})\Big)\\
        &= 
        \covlast_{n}\,\covlast_n^{-1}\,\mlast_{n-1}
        +\covlast_{n-1}\,F_{n, \plast}\,\Big(\yt[n] - \sigma(\bar{f}_{n-1})\Big)\\
        &= \mlast_{n-1} + \covlast_{n-1}\,F_{n, \plast}\Big(\yt[n] - \sigma(\bar{f}_{n-1})\Big)\\
        &= \mlast_{n-1} +\covlast_{n-1}\, h(\mhidden_{n-1}; \xt[n])  \,\Big(\yt[n] - \sigma(\mlast_{n-1}^\intercal\, h(\mhidden_{n-1}; \xt[n]))\Big)\,.
    \end{align*}
    The updates for $\mhidden_n$ and $\covhidden_n$ are obtained similarly.
\end{proof}

\section{Robustness analysis}\label{appendix:robustness-checks}
\subsection{Inventory aversion}\label{app: amb aversion}
Here, we study how the value of inventory aversion parameter $\Phi$ affects the the AUC and the internalisation-externalisation strategy.
Intuitively, as the value of inventory aversion parameter increases, the broker is willing to increase the {\cutoffprobability} to internalise a trade if it would reduce the absolute value of her inventory;
similarly, she is willing to decrease the {\cutoffprobability} to internalise a trade if internalising the trade increases the absolute value of her inventory.

A higher value of the inventory aversion parameter does not necessarily mean that the broker externalises more trades,
it means that the broker is keener to externalise trades that decrease the absolute value of her inventory and less keen to internalise trades
that increase the absolute value of her inventory.
Figure \ref{fig: app-int-vol-amb} (left) shows the internalised volume as a function of the inventory aversion parameter $\Phi$.
Figure \ref{fig: app-int-vol-amb} (right) shows the PnL and avoided profits
of the internalisation-externalisation strategy in \eqref{eq: optimal strategy delta^pm}
as a function of the inventory aversion parameter $\Phi$ for a {\toxichorizon} of 60s using {\ourmodel}.

\begin{figure}[H]
    \centering
    \includegraphics[width=0.8\linewidth]{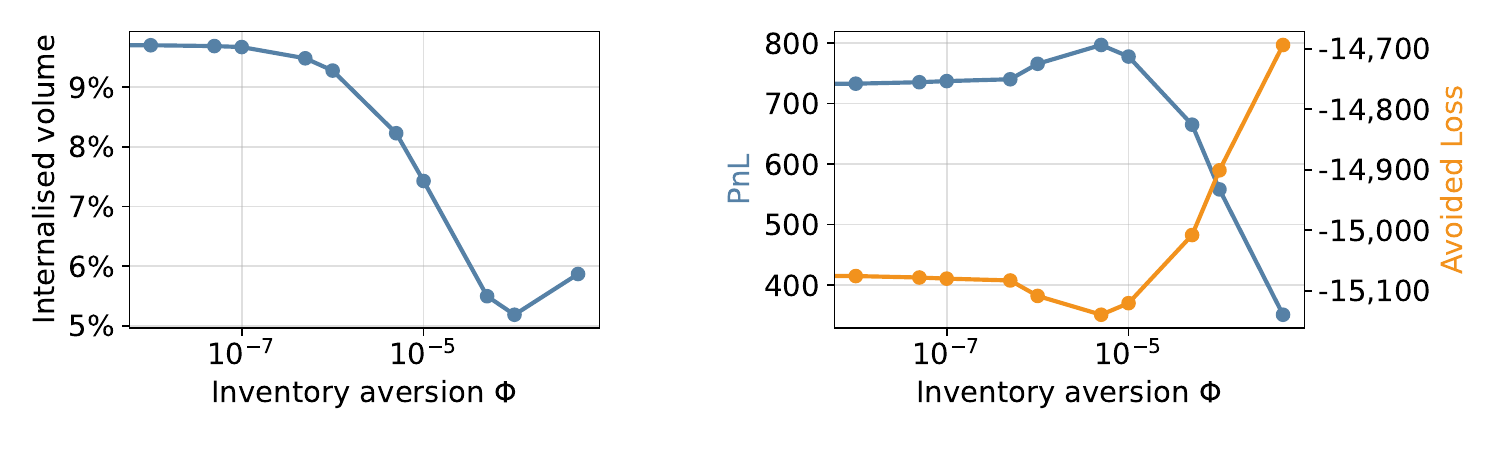}
    \caption{
    Left panel: percentage of internalised volume as a function of the inventory aversion parameter $\Phi$. Right panel: PnL and avoided profit of the internalisation-externalisation strategy as  a function of the inventory aversion parameter $\Phi$.
    }
    \label{fig: app-int-vol-amb}
\end{figure}
As the value of the inventory aversion parameter $\Phi$ increases,
the broker tends to internalise fewer trades. 
The performance of the internalisation–externalisation strategy is stable for small values of $\Phi$, 
but when $\Phi$ becomes large the strategy performs poorly. 
{
To see this, observe that when $\Phi>0$,  a positive (negative) inventory $Q>0$ ($Q<0$) makes the effective cutoff probability  higher for trades that increase (decrease) the inventory further; this makes the broker less likely to internalise trades in these respective directions. 
In our experiments, this asymmetry in the effective cutoff probability undermines the predictions made by the models
which lowers the PnL.
The skewing of effective cutoff probability also implies that fewer trades are internalised (see left panel). 
Fewer internalised trades also reduces realised PnL (on average)
because  the PnL of the broker is positive for the baseline experiments.
At the same time, avoided profit increases because more of these profitable (for the broker) trades are externalised. 
In the high-inventory aversion regime the strategy is dominated by the urgency to unwind inventory, 
and the predictive power of the models matters less.
}

Figure \ref{fig:volume-aversion} shows the daily volume of the internalisation-externalisation strategy in \eqref{eq: optimal strategy delta^pm} as a function of the inventory aversion parameter $\Phi$ for a {\toxichorizon} of 60s using {\ourmodel}.
This is a detailed version of what the left panel of Figure \ref{fig: app-int-vol-amb} describes.
Indeed, the higher the value of the inventory aversion parameter, the less the broker internalises.

\begin{figure}[H]
    \centering
    \includegraphics[width=0.6\linewidth]{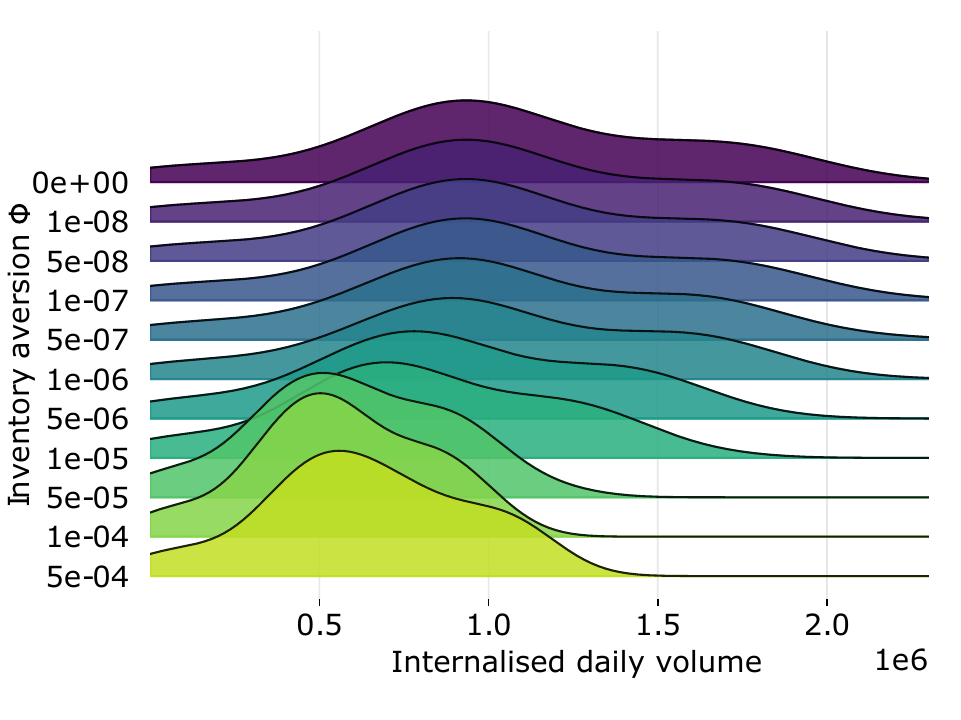}
    \caption{
    Kernel density estimate plot of daily traded volume for a range of values of the inventory aversion parameter $\Phi$.
    }
    \label{fig:volume-aversion}
\end{figure}

Finally, Figure \ref{fig:pct-int-vol-by-toxic-horizon-all}
shows the percentage of internalised volume as a function of the  {\cutoffprobability}
and the  {\toxichorizon}.
\begin{figure}[H]
    \centering
    \includegraphics[width=0.8\linewidth]{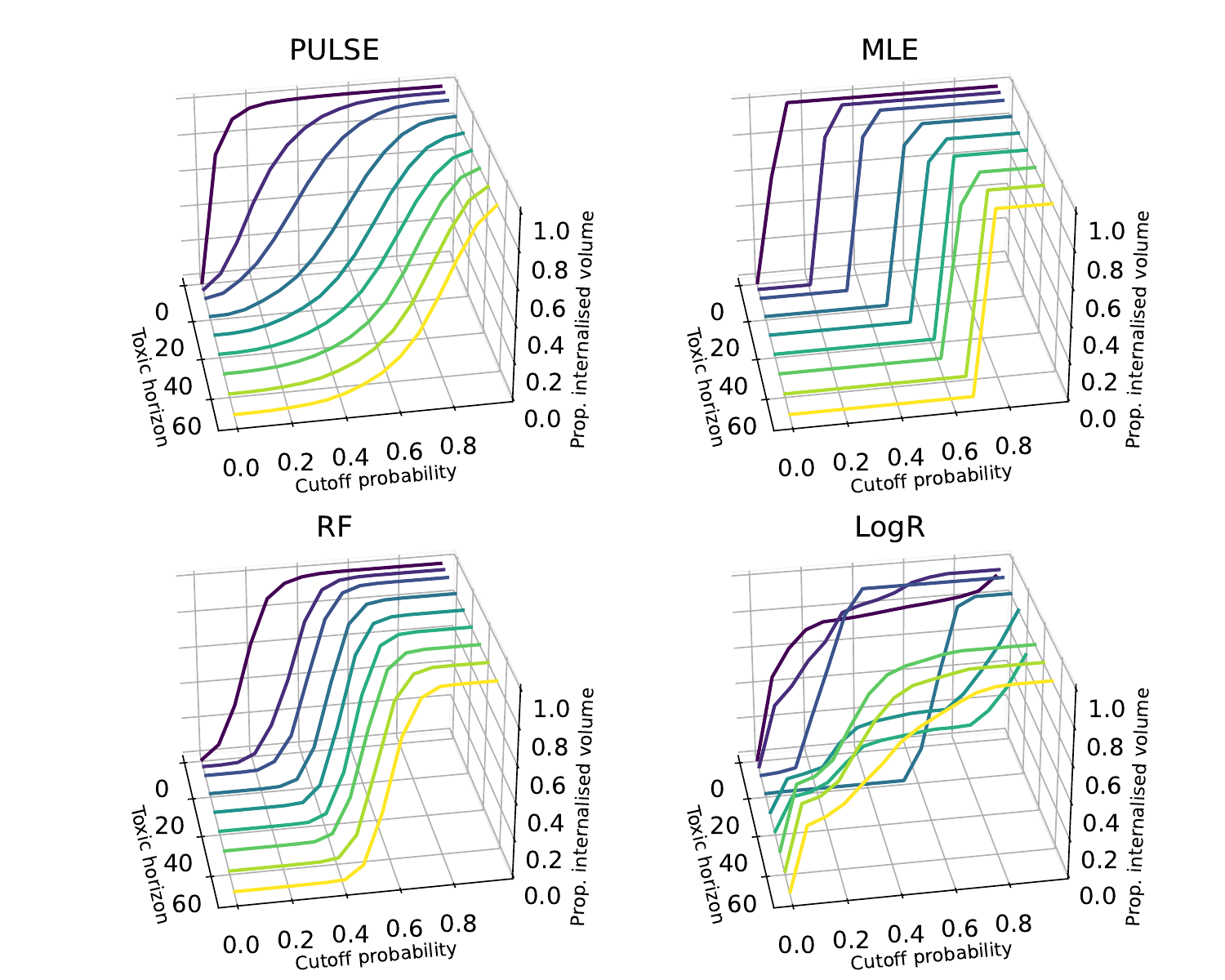}
    \caption{
        Proportion of internalised volume as a function of the  {\cutoffprobability} $\cutoff$
        and {\toxichorizon}.
        EUR/USD currency pair over the period 1 August 2022 to 21 October 2022.
    }
    \label{fig:pct-int-vol-by-toxic-horizon-all}
\end{figure}

\subsection{Performance for independent clocks}
\label{sec:independent-clock-performance}
In Subsection \ref{sec:features} we constructed 168 of the 175 features with three clocks: transaction-clock, time-clock, and volume-clock.
Here, we explore the AUC of our models 
when we use only one of the three clocks to construct the features, we omit MLE because its predictions are feature-free.
Recall that 168 out of the 183 features used in the calculations are measurements of 8 variables with three clocks and seven time horizons
($168 = 3\times 8\times 7$, see Subsection \ref{sec:features}).
Here, instead, we evaluate the model with 56 clock-features instead of 168, that is,
we use either
(i) the transaction-clock,
(ii) the time-clock,
or (iii) the volume-clock.
Table \ref{tab:auc-all} shows the AUCs, where
we observe that the models are robust to the choice of clock.

\begin{table}[H]
    \centering
\begin{tabular}{l|c|ccc}
  \noalign{\vskip 1mm}
  \hline\hline
  \noalign{\vskip 1mm}
  & all  & time & txn & vol  \\
 \noalign{\vskip1mm}
 \hline
 \noalign{\vskip-1mm}
 \noalign{\vskip 2mm}
PULSE & 62.5 & 62.6 & 62.6 & 62.6 \\ 
LogR & 50.0 & 50.0 & 50.0 & 50.0 \\
RF & 56.8 & 56.3 & 53.8 & 53.1 \\
\noalign{\vskip 1mm}
  \hline
  \hline
\end{tabular}
    \caption{AUC for {\toxichorizon} of 30s by model and by
    clock.
    EUR/USD currency pair over the period 1 August 2022 to 21 October 2022.
    Here,
    \textit{time} is the time-clock,
    \textit{txn} is the transaction-clock, and
    \textit{vol} is the volume clock.
    }
    \label{tab:auc-all}
\end{table}

We conclude that for {\ourmodel}, considering all three clocks does not add value. For RF on the other hand, the extra feature-engineering exercise does add value.

\subsection{The added value of one model per client}\label{app-section:value-of-single-model}
Here, we study how model predictions change when we fit a model per client as opposed to one model for all clients.
Figure \ref{fig:heatmap-accuracy-outperformance} shows the outperformance in AUC from an individual model over the universal model (one model for all clients).

\begin{figure}[H]
    \centering
    \includegraphics[width=0.7\linewidth]{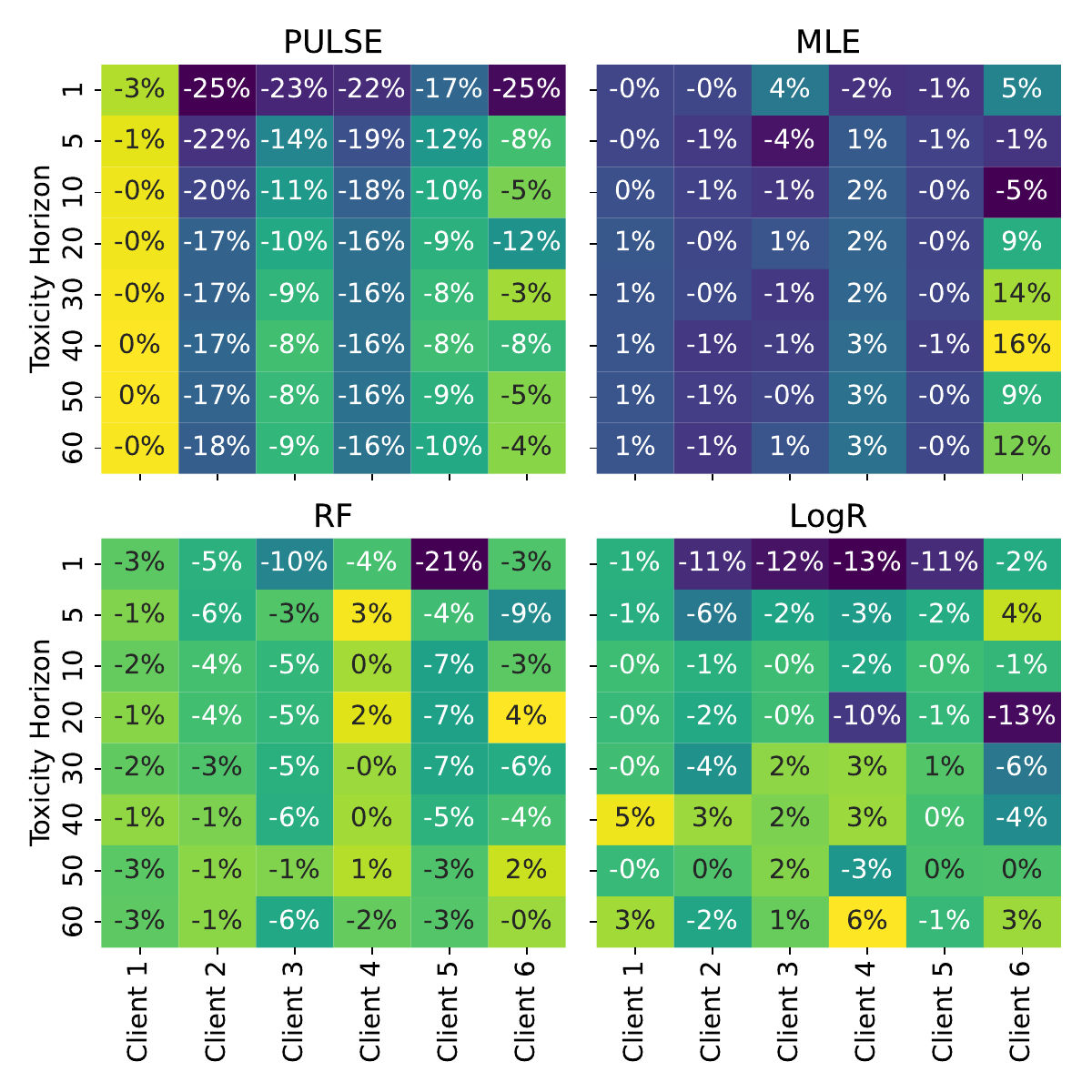}
    \caption{
        Ability to predict toxic flow.
        Difference between AUC measure without and with trader identification.
    }
    \label{fig:heatmap-accuracy-outperformance}
\end{figure}

For {\ourmodel}, with the exception of Client 1, the performance of one model per client is significantly lower than that of the universal model. Indeed, the additional data are more valuable than a model per client. 
The universal model employs features that are built using the identity of the client, e.g., inventory and cash of client. Thus, it is more advantageous to have one model for all clients and benefit from more data points than one model per client at the expense of having fewer data points to train the individual models. This is the case for {\ourmodel} and RF, whereas the results for LogR and MLE are not conclusive.

The above results are for a universal model that does employ client-specific features. If we consider a universal model without access to client-specific features (e.g., cash, inventory, and recent activity) the performance of the models deteriorates. For example, for Client 1 (using PULSE) we find close to 20\% decrease in the ability to predict toxic flow when going from a universal model with client-specific features to a universal model without client-specific features.

\subsection{Global models}
In this section we study the added value of employing data from more clients to build the models.
{
Figure~\ref{fig:auc-bday-btoxic-horizon-top100}
reports the median daily AUC across toxicity horizons under two training regimes:
(i) a single \emph{global} model and (ii) \emph{client-specific} models (one model per client).
The global model is fit on pooled data from the top 100 clients and excludes client-identifying features.
The client-specific models are trained separately for the top six clients and include client information features.
}

\begin{figure}[H]
    \centering
    \includegraphics[width=0.8\linewidth]{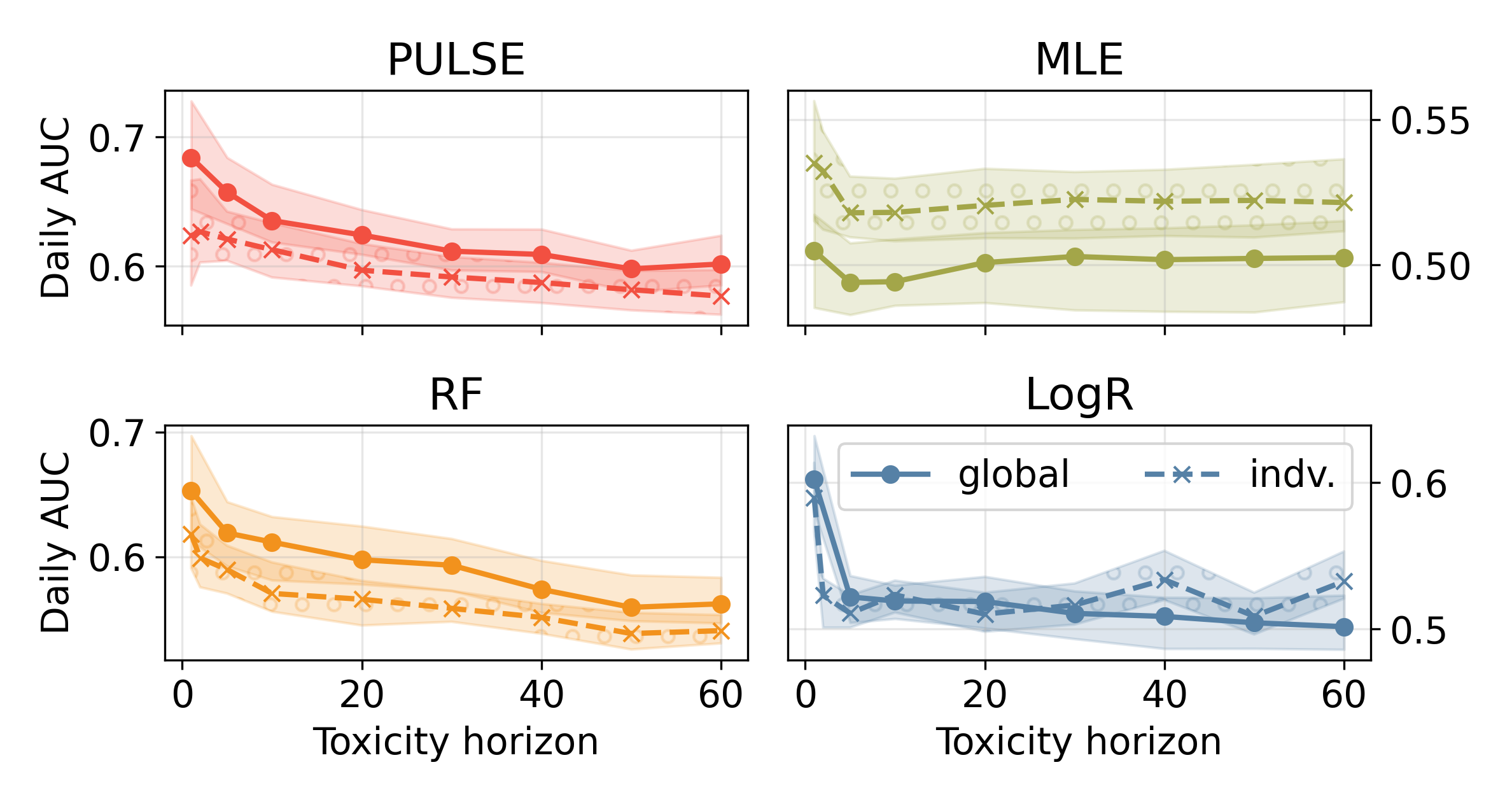}
    \caption{
        {
        Each panel shows daily AUC versus toxicity horizon for one method under two configurations:
        (i) per-client models and (ii) a single global model.
        Solid lines denote the global configuration; dashed lines denote per-client;
        the bands around each line indicate the interquartile range (IQR).
        EUR/USD, 1 Aug 2022–21 Oct 2022.
        }
    }
    \label{fig:auc-bday-btoxic-horizon-top100}
\end{figure}
{
Pooling data helps flexible models.
The mean daily AUC is higher for {\ourmodel} and RF when using a single \emph{global} model than when fitting one model per client---these
methods can exploit the larger sample and share statistical strength across clients.
LogR performs roughly on par in both settings, though the per-client variant shows higher variance due to smaller sample sizes.
MLE is the exception:
a per-client MLE (client-specific base rate) outperforms the global MLE,
since a global average dilutes client idiosyncrasies.
Overall, {\ourmodel} is the top performer, followed by RF
for both training regimes.
}
Thus, in our dataset, having more clients (more data) is advantageous.
This finding is specific to our dataset (it is not universal);
this might be a consequence of the pool of clients we study,
the features we consider, or
the period in our analysis.

\subsection{Weekly re-trains}
\label{sec:weekly-retrains}
In this section, we compare the performance of the NNet trained with {\ourmodel}
to the performance of the RF and LogR models with weekly re-training.
At the beginning of each week of the deploy phase,
we use the previous week's data to re-train RF and LogR.
For the first week in the deploy phase, we make predictions with data of the warmup phase.
Figure \ref{fig:models-with-weekly-retrain-comparison} shows the
five-day rolling AUC mean with $\valuetoxichorizon = 10s$.
\begin{figure}[htb]
    \centering
    \includegraphics[width=0.6\linewidth]{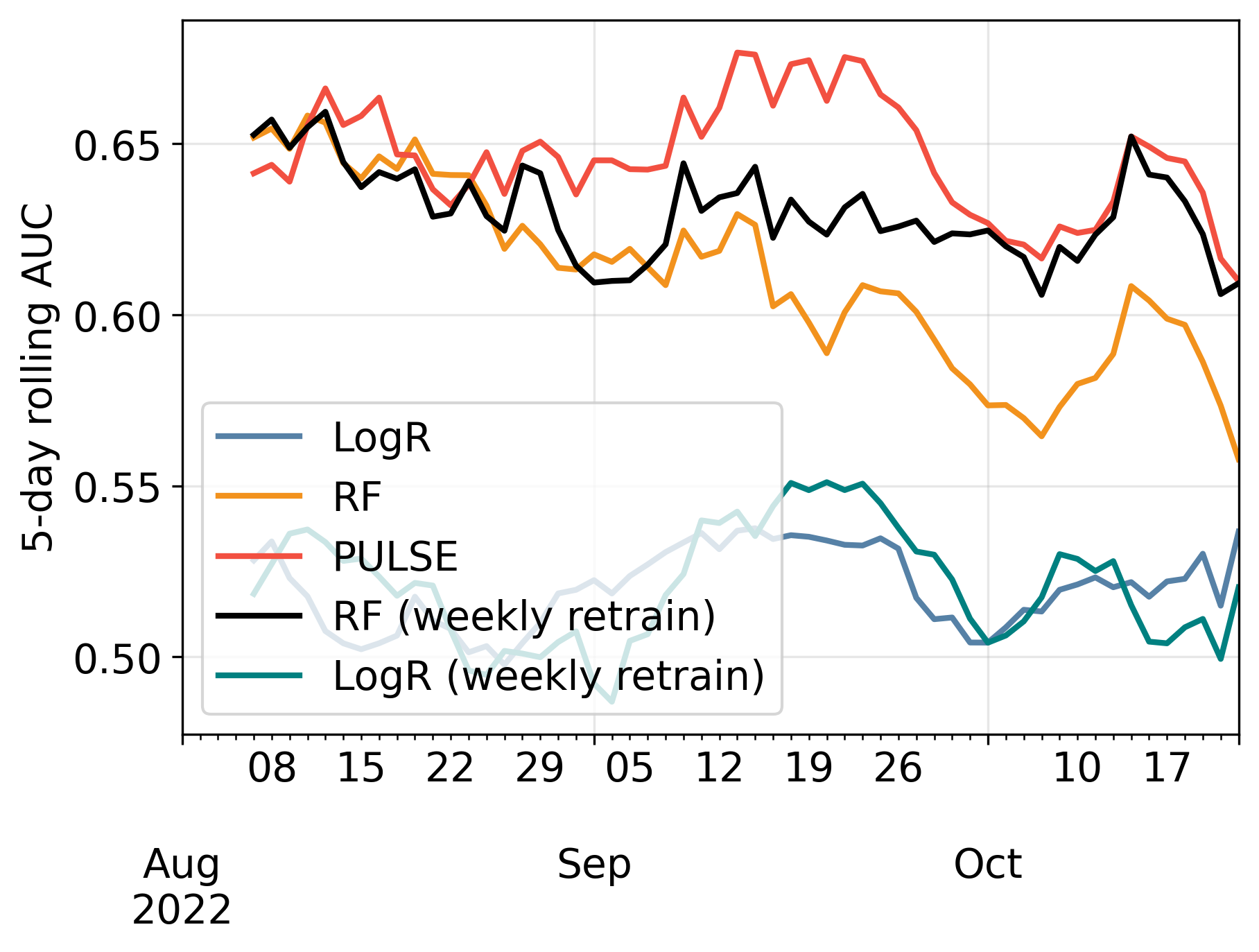}
    \caption{ Five-day exponentially-weighted moving average of AUC over time. The {\toxichorizon} is ten seconds.}
    \label{fig:models-with-weekly-retrain-comparison}
\end{figure}
The decrease in performance for RF and LogR with weekly retrain is less pronounced than when 
training only in the deploy phase.
As above, the mean AUC of RF and LogR with weekly re-retraining is lower than that of the NNet trained with {\ourmodel}.

The NNet trained with {\ourmodel} only observes each datapoint once, whereas 
RF keeps the data in memory  to re-train.
Thus, a NNet trained with {\ourmodel} is more memory-efficient  and more scalable than both RF and LogR.
Finally, given that {\ourmodel} is fully-online, it reacts more quickly to changes in the behaviour of 
clients.

\subsection{Variance of LogR}
\label{sec:variance-logR}
{
Here,  we investigate the variance we  observed in Figure \ref{fig:auc-bday-btoxic-horizon},
Figure \ref{fig:pct-int-vol-by-toxic-horizon-all}, and
 the concentrated probabilities shown in Figure \ref{fig:predicted-probas}.
}

{
Figure \ref{fig:normed-coeff-logr} shows the normalised coefficient magnitudes $|w_j|/\sum_k |w_k|$ for per-client LogR fits. 
We observe a marked concentration of mass: few coefficients (in some instances, one) dominate,
while the remainder are near zero, and the index of the dominant coefficient varies across clients. 
This reflects weak linear signal and collinearity in small per-client samples.
Consistent with this, LogR's predicted probabilities concentrate near a base rate,
yielding near-step internalised-volume curves and the horizon-to-horizon jaggedness noted in the main text.
}
\begin{figure}[H]
    \centering
    \includegraphics[width=0.6\linewidth]{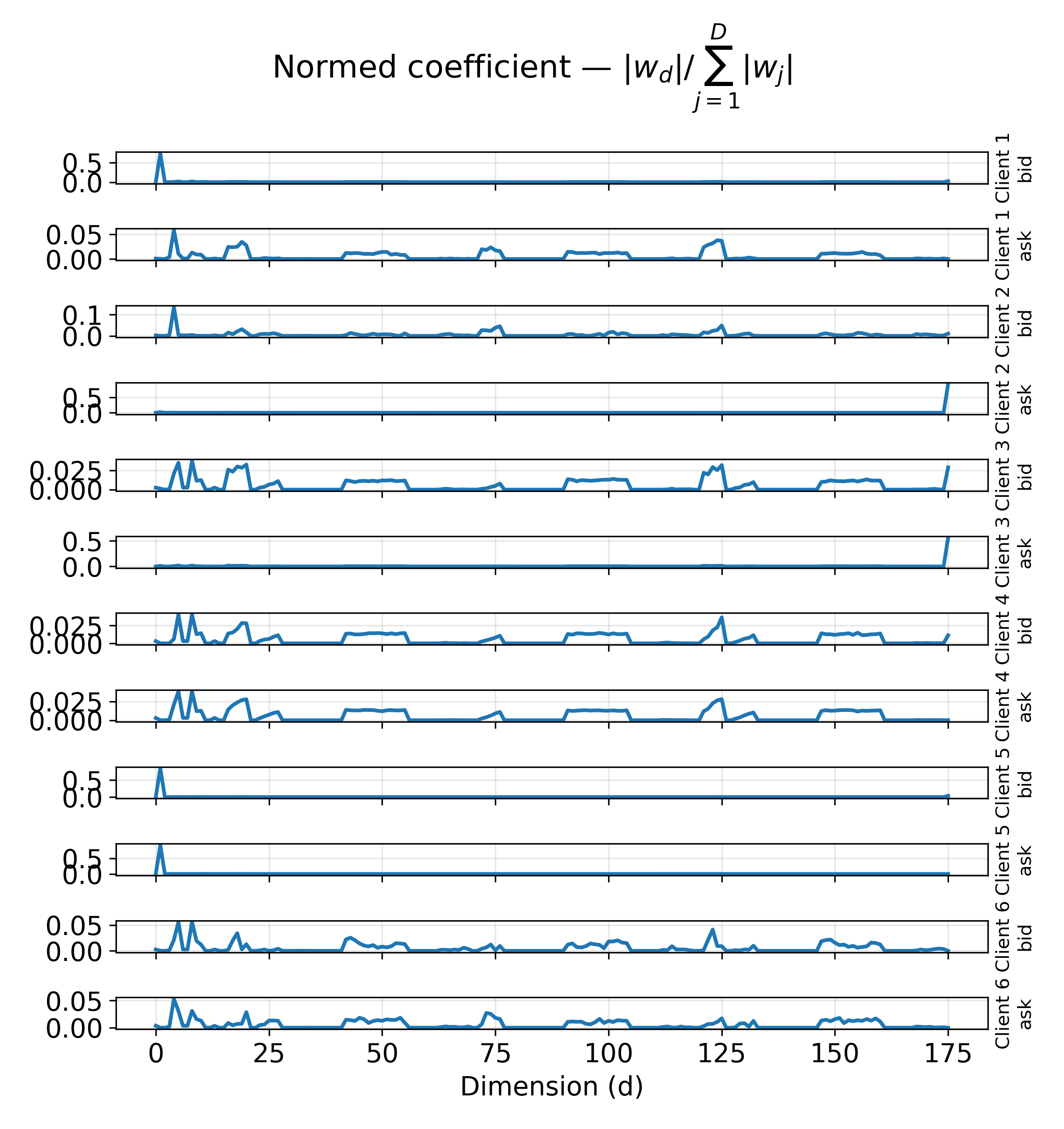}
    \caption{
        {
            Normalised coefficient magnitudes $|w_d|/\sum_j |w_j|$ for per-client LogR fits (top six clients). 
        }
    }
    \label{fig:normed-coeff-logr}
\end{figure}

\subsection{Precision and recall}\label{sec: precision recall}

{
Figure~\ref{fig:precision-recall} plots precision and recall versus the cutoff probability at a {\toxichorizon} of 20s.
LogR and MLE exhibit flat precision and rapidly declining recall, reflecting limited score dispersion.
RF and {\ourmodel} show clearer trade-offs, with precision increasing as the cutoff rises.
These patterns are consistent with the internalised-volume behaviour in Figure~\ref{fig:pct-int-vol} and with the AUC comparisons in Figure~\ref{fig:auc-bday-btoxic-horizon}.
}

\begin{figure}[H]
    \centering
    \includegraphics[width=0.7\linewidth]{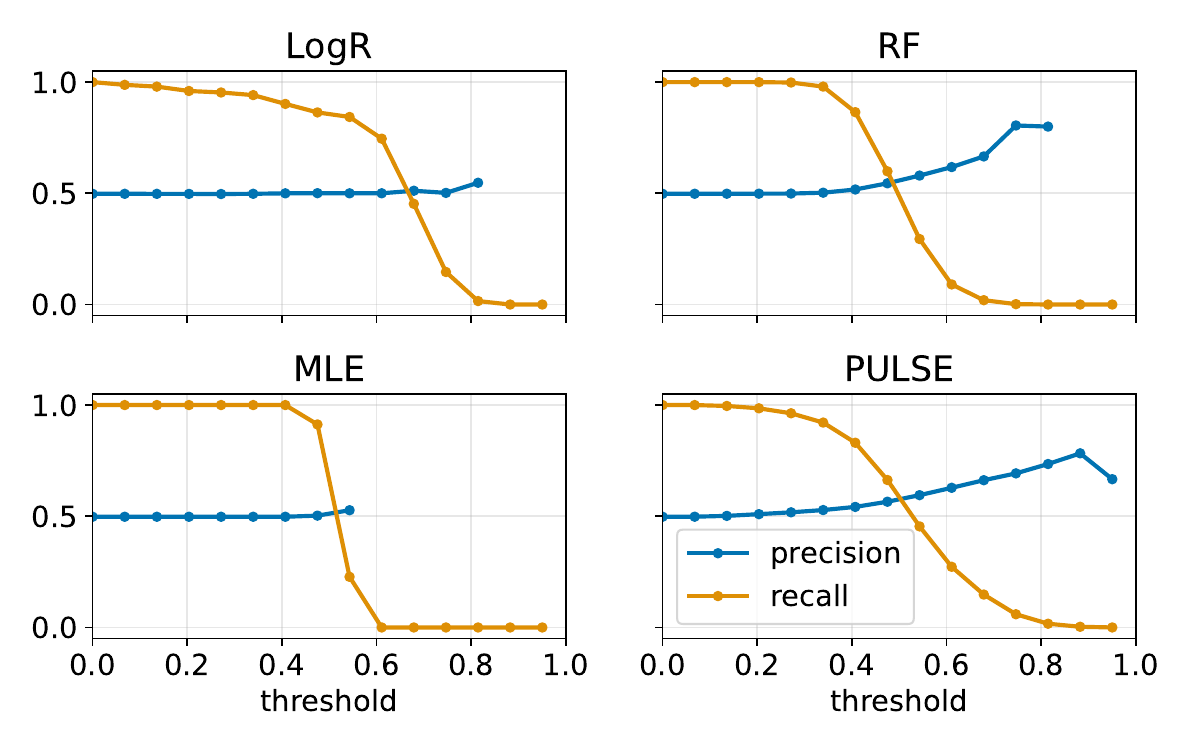}
    \caption{
    Precision and recall as a function of the cuttoff probability \cutoffprobability at a 20s toxicity horizon, by model.
    }
    \label{fig:precision-recall}
\end{figure}

\end{document}